\documentclass[journal]{IEEEtran}

\usepackage{amsmath,amssymb}
\usepackage{graphicx,graphics,color,psfrag}
\usepackage{cite,balance}
\usepackage{algorithm}
\usepackage{accents}
\usepackage{bm}
\usepackage{url}
\usepackage{algorithmic}
\usepackage[english]{babel}
\usepackage{multirow}
\usepackage{enumerate}
\usepackage{cases}
\usepackage{stfloats}
\usepackage{dsfont}
\usepackage{color,soul}
\usepackage{amsfonts}
\usepackage{tcolorbox}
\usepackage{amsmath}
\usepackage{float}
\usepackage{mathtools}
\usepackage{verbatim} 
\usepackage{bm,color,soul}
\usepackage{epsfig}
\usepackage{subfigure}
\usepackage{psfrag}
\usepackage{latexsym}
\usepackage[footnotesize]{caption}
\usepackage{color}
\usepackage{booktabs}
\usepackage[caption=false,font=normalsize,labelfont=sf,textfont=sf]{subfig}

\IEEEoverridecommandlockouts

\usepackage{algorithm}
\usepackage{algorithmic}
\usepackage{cuted}
\usepackage{lipsum}
\usepackage{stfloats}
\usepackage{cite,graphicx,amsmath,amssymb}
\usepackage{fancyhdr}
\usepackage{hhline}
\usepackage{array,color}
\usepackage{amsmath}
\usepackage{stfloats}
\usepackage[flushleft]{threeparttable}
\usepackage{makecell} 
\newtheorem{prop}{Proposition}

\newtheorem{proof}{proof}

\allowdisplaybreaks[1] 

\begin{document}

%
\title{Networked ISAC for Low-Altitude Economy: Coordinated Transmit Beamforming and UAV Trajectory Design}



\author{Gaoyuan~Cheng, Xianxin~Song, Zhonghao~Lyu, and  Jie~Xu
\thanks{An earlier version of this paper has been accepted at the 2024 IEEE/CIC International Conference on Communications in China (ICCC) \cite{cheng2024LAE}.
}
\thanks{G. Cheng, X. Song, Z. Lyu, and J. Xu are with the School of Science and Engineering (SSE), the Shenzhen Future Network of Intelligence Institute (FNii-Shenzhen), and the Guangdong Provincial Key Laboratory of Future Networks of Intelligence, The Chinese University of Hong Kong, Shenzhen, Guangdong 518172, China (e-mail: gaoyuancheng@link.cuhk.edu.cn, xianxinsong@link.cuhk.edu.cn, zhonghaolyu@link.cuhk.edu.cn, xujie@cuhk. edu.cn). J. Xu is the corresponding author.}
}


\maketitle


\begin{abstract}
This paper exploits the networked integrated sensing and communications (ISAC) to support low-altitude economy (LAE), in which a set of networked ground base stations (GBSs) cooperatively transmit joint information and sensing signals to communicate with multiple authorized unmanned aerial vehicles (UAVs) and concurrently detect unauthorized objects over the interested region in the three-dimensional (3D) space. We assume that each GBS is equipped with uniform linear array (ULA) antennas, which are deployed either horizontally or vertically to the ground. We also consider two types of UAV receivers, which have and do not have the capability of canceling the interference caused by dedicated sensing signals, respectively. Under each setup, we jointly design the coordinated transmit beamforming at multiple GBSs together with the authorized UAVs' trajectory control and their GBS associations, for enhancing the authorized UAVs' communication performance while ensuring the sensing requirements. In particular, we aim to maximize the average sum rate of authorized UAVs over a given flight period, subject to the minimum illumination power constraints toward the interested 3D sensing region, the maximum transmit power constraints at individual GBSs, and the flight constraints of UAVs. These problems are highly non-convex and challenging to solve, due to the involvement of binary UAV-GBS association variables as well as the coupling of beamforming and trajectory variables. To solve these non-convex problems, we propose efficient algorithms by using the techniques of alternating optimization, successive convex approximation, and semi-definite relaxation. Numerical results show that the proposed joint coordinated transmit beamforming and UAV trajectory designs efficiently balance the sensing-communication performance tradeoffs and significantly outperform various benchmarks. It is also shown that the horizontally placed antennas lead to enhanced performance compared with their vertical counterparts due to the more flexible multi-beam design, and the sensing interference cancellation ability at UAV receivers is advantageous for further enhancing ISAC performance.

\end{abstract}



\section{Introduction}
Low-altitude economy (LAE) corresponds to a comprehensive economic form consisting of various low-altitude flight activities of unmanned and manned aircraft, such as unmanned aerial vehicles (UAVs) and electric vertical take-off and landing (eVTOL). It is envisioned that LAE can enable a series of low-altitude applications in transportation, cargo delivery, entertainment, environmental monitoring, agriculture, and public security to create great economic and social value, which has attracted explosively increasing research attention around the world recently \cite{Whitepaper4}. The successful implementation of LAE, however, requires the safe operation of various aircraft. It is essential to provide seamless wireless communication connections and ubiquitous sensing for massive aircraft at the low altitudes, thus extending their serving region, supporting their trajectory planning and tracking, and providing real-time monitoring to prevent invasion of unauthorized objects \cite{Mu2023magazine, Meng2023magazine}.

As one of the key technologies for six-generation (6G) wireless networks, integrated sensing and communications (ISAC) has emerged as an efficient solution to support LAE \cite{Wu2021Comprehensive, Mozaffari2019Tutorial, mozaffari2021toward, liu2022integrated}. With ISAC, ground base stations (GBSs) can transmit wireless signals to communicate with authorized aircraft as aerial users, and reuse the reflected echo signals to sense low-altitude airspace and monitor the invasion of unauthorized objects \cite{jiang20236g,  Cui2023Specific, Hu2022Trajectory}. More specifically, thanks to the inter-connected nature of distributed GBSs, networked ISAC is particularly appealing to provide large-scale sensing and communication services for LAE, in which distributed GBSs can cooperate in not only coordinated multi-point (CoMP) transmission and reception, for  communication \cite{Yu2010Coordinated, Gesbert2010Multi, Nigam2014Coordinated} but also distributed multiple-input multiple-output (MIMO) radar for sensing \cite{Fan2023Sensing, zhang2021perceptive, fishler2006spatial, haimovich2007mimo}. Compared with conventional ISAC systems with isolated transceivers, networked ISAC offers several advantages. First, networked ISAC significantly extends the coverage of sensing communication by GBS cooperation, which is particularly useful for LAE with massive aircraft distributed in large-scale 3D space. Next, networked ISAC can enable cooperative transmission and reception among GBSs via implementing joint signal processing at the central processors \cite{Cheng2024Optimal, Ngo2017Cell-Free, Wu2015Cloud}. This allows the GBSs to properly mitigate and even utilize the multi-cell air-ground interference in sensing and communication. Furthermore, GBSs are distributed at separate locations, which can view the same target from different angles and exploit the spatial diversity of target radar cross section (RCS) to enhance the sensing performance \cite{Fei2023magazine, huang2022coordinated, Godrich2010Target, Lehmann2006High}. To fully reap the above benefits, it is of utmost importance to exploit networked ISAC to support real-time communication and tracking of authorized aircraft, and provide seamless monitoring of targeted regions in the 3D space.

In the literature, there have been various prior works investigating the use of terrestrial cellular networks to support the communication and sensing of UAVs. On the one hand, cellular-connected UAV has been extensively studied \cite{Zeng2019CellularUAV}, in which GBSs support the communication of UAVs as aerial users by using the spectrum resources originally allocated to terrestrial subscribers. In particular, existing works studied cellular-connected UAV from different perspectives such as interference mitigation \cite{peiming2022Cellular, Mei2019Cellular}, energy-efficient communications\cite{Zeng2017Energy, Zhan2020EnergyUAV}, and mobile edge computing \cite{Cao2018Mobile}. On the other hand, some recent works \cite{Cui2023Specific, Hu2022Trajectory} studied the use of GBS to provide ISAC services with cellular-connected UAVs. For instance, the work \cite{Cui2023Specific} considered a system with one GBS and multiple UAVs, in which a dual identity association-based ISAC method was developed to support swift beam alignment of multiple UAVs. The authors in \cite{Hu2022Trajectory} exploited the cellular-connected UAV to enable bi-static sensing  with GBS. In these above works with cellular-connected UAVs \cite{Cui2023Specific, Hu2022Trajectory, Zeng2019CellularUAV, peiming2022Cellular, Mei2019Cellular, Zeng2017Energy, Zhan2020EnergyUAV, Cao2018Mobile}, the UAV trajectory control is introduced as a new design degree of freedom (DoF) to enhance the desired signal strength and mitigate the undesired interference, thus improving the communication and sensing performances. However, these works only considered communications with cellular-connected UAVs \cite{Zeng2019CellularUAV, Mei2019Cellular, Zeng2017Energy, Zhan2020EnergyUAV, Shuowen2019Cellular}, or only investigated mono-static or bi-static ISAC for cellular-connected UAVs with one single GBS. The research on using terrestrial networked ISAC systems to support LAE has not been investigated in the literature yet.

It is worth noting that there is another line of research on UAV-assisted ISAC, in which UAVs are exploited as aerial platforms such as aerial base stations (ABSs) to provide air-to-ground ISAC services towards ground users and terrestrial objects \cite{Jing2024ISAC, Weijie2023UAV, Meng2023Throughput, Chien2024Joint, Wu2023uav, Lyu2022Joint, wang2020constrained}. For instance, the authors in \cite{Lyu2022Joint} employed a UAV-enabled ISAC platform to communicate with terrestrial users and sense interested targets in pre-determined areas, in which a joint transmit beamforming and UAV trajectory design algorithm was proposed to properly balance the tradeoff between average communication rate and sensing beampattern gains. The work \cite{wang2020constrained, Ding2023Multi} further investigated the scenario with multi-UAV-enabled ISAC towards multiple ground users and targets. However, these prior works are not applicable for terrestrial ISAC supporting LAE, in which authorized UAVs (communication users) and unauthorized objects (sensing targets) are located at the low altitude in the three-dimensional (3D) space, thus making the coverage and ISAC transmit design more challenging.

Different from prior works, this paper studies the exploitation of networked ISAC to support LAE. In particular, we consider a scenario with multiple multi-antenna UAVs cooperatively designing the transmit beamforming to provide seamless communication with authorized UAVs and real-time monitoring of intended areas in the 3D space. We also exploit the trajectory design of authorized UAVs to optimize the communication performance subject to the sensing requirements. However, the efficient design of cooperative transmit beamforming at GBSs and trajectory optimization of authorized UAVs is particularly challenging. First, the sensing-communication performance tradeoff is dependent on both the transmit beamforming and the authorized UAVs' 3D locations.  In particular, GBSs may steer their transmit beams between sensing areas and UAV locations to balance the sensing and communication requirements, and authorized UAVs can fly closely to the areas with strong beamforming gains to enjoy enhanced communication performance. It is thus important but difficult to jointly adapt both transmit beamforming and UAV trajectory. Next, the simultaneous transmission of multiple GBSs introduces inter-GBS interference, while their ISAC operation also leads to interference between sensing and communication. Due to the line-of-sight (LoS) dominated air-ground wireless channels, the interference is rather strong, thus making interference management challenging. Furthermore, due to the mobility of authorized UAVs, their association relationship with GBSs may change over time, which may significantly affect the interference management and the resultant performance. It is thus necessary to further optimize the binary UAV-GBS association together with transmit beamforming and UAV trajectory control, thus making the design problem even more difficult. Therefore, we are motivated to address the above challenges in this work.

In particular, we consider that a set of networked GBSs with uniform linear array (ULA) antennas cooperatively communicate with multiple authorized UAVs and simultaneously sense an interested area in the 3D space to monitor unauthorized objects. The main results are summarized in the following:
\begin{itemize}

\item We consider that each GBS sends joint information and dedicated sensing signals to facilitate the ISAC operation. Under this joint signal design, we consider two types of UAV receivers, i.e., Type-I and Type-II UAV receivers, which have and do not have the capability to cancel the interference from dedicated sensing signals, respectively. Furthermore, to analyze the effect of antenna configuration on the ISAC performance, we consider two different cases when the ULA antennas at each GBS are deployed horizontally and vertically to the ground, respectively.

\item Under each of the above setups, we jointly optimize the coordinated transmit beamforming of GBSs, the authorized UAVs' trajectory control, and the GBS-UAV association, with the objective of maximizing the average sum rate of authorized UAVs over a particular ISAC period, subject to the transmit power constraints at GBSs, the practical flight constraints at UAVs, and the minimum illumination power constraints for sensing over the targeted 3D region. However, the formulated problem is difficult to solve due to the involvement of integer variable constraints and the coupling of beamforming and trajectory variables. To address these issues, we propose an efficient algorithm to find a high-quality solution by applying the techniques of alternating optimization (AO),  semidefinite relaxation (SDR), and successive convex approximation (SCA). The convergence of the proposed algorithm is ensured.

\item Finally, numerical results are provided to validate the performance of our proposed designs as compared to benchmark schemes with straight flight and isotropic transmission, respectively. It is shown that the proposed joint coordinated transmit beamforming and UAV trajectory design efficiently balances the tradeoff between sensing and communication performance and significantly outperforms the benchmark schemes. It is also shown that the horizontally placed antennas at GBSs lead to better ISAC performance than the vertically placed case. Furthermore, for both antenna configurations, the Type-II receivers with the ability to cancel the sensing interference are shown to outperform the Type-I receivers in terms of average sum rate.

\end{itemize}

The remainder of this paper is organized as follows. Section \ref{sec:system} presents the networked ISAC model for LAE and formulates the problem. Section \ref{Sec:solution} presents the proposed joint coordinated transmit beamforming and UAV trajectory design in the case with horizontally placed antennas at GBSs and Type-I UAV receivers. Section \ref{sec:other} presents the proposed designs for other cases of antenna configuration and UAV receivers. Section \ref{sec:numerical} provides numerical results to illustrate the efficiency of our proposed designs. Section \ref{sec:conclusion} concludes this paper.


{\textit{Notations:}} Lowercase and uppercase letters with boldface refer to vectors and matrices, respectively. $\mathbb E (\cdot)$ means the statistical expectation. For an arbitrary scalar $a$ and an arbitrary vector $\bf a$, $\left| a \right|$ and $\left\| {\bf{a}} \right\|$ denote the absolute value of $a$ and Euclidean norm of $\bf a$, respectively. The superscripts $T$ and $H$ denote the transpose and conjugate transpose operators for matrices and vectors. ${\mathbb C}^{x \times y}$ denotes the space of $x \times y$ complex matrices. $j=\sqrt{-1}$ denotes the imaginary unit.



\begin{figure}[ht]
\centering
    \includegraphics[width=7cm]{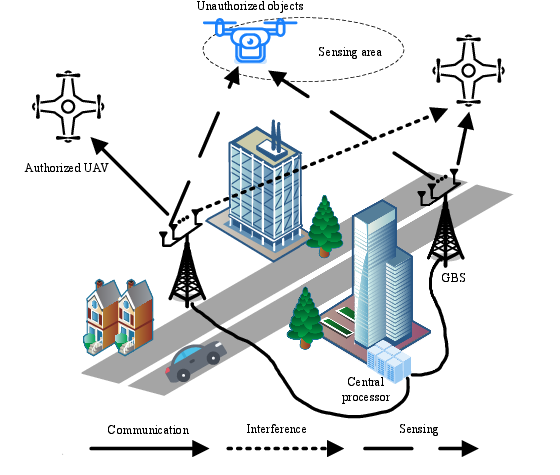}
\caption{Illustration of networked ISAC for communicating with authorized UAVs and monitoring of unauthorized objects in 3D space.}
\label{fig:system}
\end{figure}

\section{System Model and Problem Formulation}\label{sec:system}
We consider a networked ISAC system as illustrated in Fig. \ref{fig:system}, which consists of $M$ GBSs each  with $N_a$ antennas and $K$ authorized UAVs each with a single antenna. The sets of GBSs and authorized UAVs are denoted as ${\cal M} = \{1,\ldots,M \}$ and ${\cal K} = \{1,\ldots,K \}$, respectively. In this system, the GBSs perform downlink communication with their  associated authorized UAVs and concurrently sense the targeted 3D area for monitoring unauthorized objects.

We focus on the networked ISAC operation over a particular time period ${\cal T}=[0, T]$ with duration $T$, which is divided into $N$ time slots each with duration ${\Delta _t} = T/N$. Let ${\cal N}=\{1, \ldots, N\}$ denote the set of slots. Here, $N$ is chosen to be sufficiently large and accordingly ${\Delta _t}$ is sufficiently small, such that the UAVs' locations are assumed to be unchanged over each slot to facilitate the system design. Without loss of generality, we consider a 3D Cartesian coordinate system. Let ${\bf u}_m=(x_m,y_m)$ denote the horizontal coordinate of each GBS $m \in {\cal M}$, and ${\bf q}_k[n]= ({\hat x}_k[n], {\hat y}_k[n])$ denote the time-varying horizontal coordinate of UAV $k \in {\cal K}$ at slot $n \in {\cal N}$. All GBSs are located at the zero altitude, and each UAV $k \in {\cal K}$ is at a fixed altitude of $H_k>0$, where $H_k$'s can be different among different UAVs due to their pre-assigned flight region.

We consider two types of antenna configurations at GBSs, i.e., the horizontally and vertically placed ULA antennas, in which the ULA antennas at each GBS are placed parallel to the x-axis and z-axis in the 3D Cartesian coordinate system, respectively. For the horizontal and vertical antenna configuration cases, the corresponding angles of departure (AoDs) between GBS $m$ and UAV $k$ at time slot $n$ are respectively denoted as
\begin{align}
\theta^{(x)}_m ({{\bf{q}}_k}[n],H_k) = \arccos \frac{{{{\boldsymbol{\psi }}^T}({{\bf{q}}_k}[n] - {{\bf{u}}_m})}}{{\sqrt {{{\left\| {{{\bf{q}}_k}[n] - {{\bf{u}}_m}} \right\|}^2} + H_k^2} }}, \label{aod:x}
\end{align}
\begin{align}
\theta^{(z)}_m ({{\bf{q}}_k}[n], H_k) = \arccos \frac{H_k}{{\sqrt {{{\left\| {{{\bf{q}}_k}[n] - {{\bf{u}}_m}} \right\|}^2} + H_k^2} }}, \label{aod:z}
\end{align}
where ${\boldsymbol{\psi }}=[1,0]^{T}$. Let $\lambda$ denote the carrier wavelength and $d$ denote the antenna spacing. We have the steering vector ${\bf{a}}_m^{(\imath )}({{\bf{q}}_k}[n],H_k)$ as
\begin{align}
&{{{\bf{a}}}_{m}^{(\imath)}}({{\bf{q}}_k}[n],H_k)= [1,{e^{j2\pi \frac{d}{\lambda }\cos \theta^{(\imath)} ({{\bf{q}}_k}[n],H_k)}}, \ldots ,\nonumber \\
&~~~~~~~~{e^{j2\pi \frac{d}{\lambda }({N_a} - 1)\cos \theta^{(\imath)} ({{\bf{q}}_k}[n],H_k)}}]^T, {\imath}\in{\{ x, z \}}. \label{steer:1}
\end{align}
Similar as in prior work \cite{Lyu2022Joint}, we consider that the air-ground links from GBSs to UAVs are dominated by LoS channels. As a result, the channel vector between GBS $m$ and UAV $k$ at slot $n$ is denoted as
\begin{align}
{{\bf{h}}_{m}^{(\imath)}}({{\bf{q}}_k}[n],H_k)&= \sqrt {\kappa ({ {\left\| {{{\bf{q}}_k}[n] - {{\bf{u}}_m}} \right\|}^2+ {H_k^2}}) ^{ - 1}} \nonumber \\
&\cdot {{\bf{a}}_{m}^{(\imath)}}({{\bf{q}}_k}[n],H_k), \imath \in \{x,z\}
\end{align}
where $\kappa$ denotes the path loss at the reference distance of one meter.

\subsection{Communication Model}

First, we consider the communication from the GBSs to the UAVs, in which the GBSs send joint information and dedicated sensing signals to associated UAVs via coordinated transmit beamforming. At each slot $n$, each UAV is associated with one single GBS. We use a binary variable ${\alpha}_{m,k}[n] \in \{0,1\}$ to indicate the association relationship between GBS $m$ and UAV $k$ at slot $n$. Here, we have ${\alpha}_{m,k}[n]=1$ if UAV $k$ is associated with GBS $m$ at slot $n$ and ${\alpha}_{m,k}[n]=0$ otherwise. As such, we have $\sum\nolimits_{l \in {\cal M}} {{{\alpha}_{l,k}}[n]}  = 1, \forall k \in {\cal K}, n \in {\cal N}$. Let $s_{m,k}[n]$ denote the transmit information signal sent by GBS $m$ to UAV $k$ at time slot $n$, and ${\bf w}_{m,k}[n]$ denote the corresponding transmit beamforming vector. Here, we assume $s_{m,k}[n]$'s are independent and identically distributed (i.i.d.) random variables each with zero mean and unit variance. Let ${\bf s}_{m}[n] \in {\mathbb C}^{N_a \times 1}$ denote the dedicated sensing signal vector sent by GBS $m$ at time slot $n$, which can be generated as pseudorandom signals  with zero mean and covariance matrix ${{\bf{R}}_m}[n]= {\mathbb E}({{\bf{s}}_m}[n]{\bf{s}}_m^H[n]) \succeq {\bf 0}$. As a result, the transmitted signal by GBS $m$ at time slot $n$ is
\begin{align}
{{\bf{x}}_m}[n] = \sum\limits_{i \in {\cal K}} {{\bf{w}}_{m,i}}[n]{s_{m,i}}[n]  + {{\bf{s}}_m}[n], n \! \in \!{\cal N}, \label{com:sig:x}
\end{align}
and the covariance matrix of ${\bf x}_m[n]$ is
\begin{align}
{{\bf{X}}_m}[n] =\sum\limits_{i \in {\cal K}}  {\bf{w}}_{m,i}[n]{{\bf{w}}_{m,i}^H}[n] + {{\bf{R}}_m[n]}.
\end{align}

In \eqref{com:sig:x}, the dedicated sensing signals ${{\bf{s}}_m}[n]$ are jointly sent with information signals to achieve full DoFs for sensing, which may introduce additional interference. However, the dedicated sensing signals ${{\bf{s}}_m}[n]$'s are predetermined pseudorandom signals in practice, and thus can be eliminated by dedicatedly designed receivers. Hence, we consider two different types of UAV receivers based on whether they have the ability to cancel the interference caused by dedicated sensing signals, namely Type-I and Type-II UAV receivers, respectively.
\begin{figure*}
\begin{align} \label{SINR}
{\gamma ^{(\imath ,{\rm{I}})}} ({{\bf{W}}_{m,k}}[n], {{\bf{R}}_m}[n],{{\bf{q}}_k}[n]) = {\frac{{{{\rm{tr}}( {\bf{H}}_{m}^{(\imath)}({{\bf{q}}_k}[n],H_k){\bf{W}}_{m,k}[n] )} }}{{\sum\limits_{(l,i) \ne (m,k)} { {{\rm{tr}}( {\bf{H}}_{m}^{(\imath)}({{\bf{q}}_k}[n],H_k){\bf{W}}_{l,i}[n] )} }  + \sum\limits_{l \in {\cal M}} {\rm{tr}}({\bf{H}}_{m}^{(\imath)}({{\bf{q}}_k}[n],H_k) {\bf{R}}_l[n] ) + \sigma _c^2}}}.
\end{align}
\end{figure*}

\begin{itemize}
\item {\bf{Type-I receivers}} are legacy users, such that they are designed for communication only and not for ISAC systems. These receivers are not able to cancel the interference generated by dedicated sensing signals $\{{{\bf{s}}_m}[n]\}$. We define ${\bf{H}}_{m}^{(\imath)}({{\bf{q}}_k}[n],H_k)= {{\bf{h}}_{m}^{(\imath)}}({{\bf{q}}_k}[n],H_k) {{\bf{h}}_{m}^{(\imath)H}}({{\bf{q}}_k}[n],H_k)$, $\imath \in \{x,z\}$ and ${{\bf{W}}_{m,k}}[n] = {{\bf{w}}_{m,k}}[n]{\bf{w}}_{m,k}^H[n]$ with ${\bf W}_{m,k}[n] \succeq 0$ and ${\text {rank}}({\bf W}_{m,k}[n]) \le 1$. Accordingly, the signal-to-interference-plus-noise (SINR) with UAV $k$ communicating with GBS $m$ is denoted by ${\gamma ^{(\imath ,{\rm{I}})}_{m,k}} ({{\bf{W}}_{m,k}}[n], {{\bf{R}}_m}[n],{{\bf{q}}_k}[n])$ in \eqref{SINR} at the top of the next page, where index $\imath \in \{x,z\}$ indicates the antenna configuration and $\sigma ^2$ denotes the noise power.

\item {\bf{Type-II receivers}} are dedicatedly designed for ISAC systems. These receivers are able to cancel the interference generated by dedicated sensing signals $\{{{\bf{s}}_m}[n]\}$ prior to decoding the desired information signals $\{{{s}_{m,k}}[n]\}$. Accordingly, the SINR with UAV $k$ communicating with GBS $m$ is denoted by ${\gamma_{m,k} ^{(\imath ,{\rm{II}})}} ({{\bf{W}}_{m,k}}[n], {{\bf{R}}_m}[n],{{\bf{q}}_k}[n])$, which is expressed as \eqref{SINRII} at the top of the next page.
\end{itemize}

\begin{figure*}
\begin{align} \label{SINRII}
{\gamma ^{(\imath ,{\rm{II}})}} ({{\bf{W}}_{m,k}}[n], {{\bf{R}}_m}[n],{{\bf{q}}_k}[n]) = {\frac{{{{\rm{tr}}( {\bf{H}}_{m}^{(\imath)}({{\bf{q}}_k}[n],H_k){\bf{W}}_{m,k}[n] )} }}{{\sum\limits_{(l,i) \ne (m,k)} { {{\rm{tr}}( {\bf{H}}_{m}^{(\imath)}({{\bf{q}}_k}[n],H_k){\bf{W}}_{l,i}[n] )} }   + \sigma _c^2}}}.
\end{align}
\vspace{-3ex}
\end{figure*}

Accordingly, we denote ${r_{m,k}^{(\imath,\jmath)}}({{\bf{W}}_{m,k}}[n], {{\bf{R}}_m}[n],{{\bf{q}}_k}[n]) = {\log _2}(1 + {\gamma_{m,k} ^{(\imath ,\jmath)}} ({{\bf{W}}_{m,k}}[n], {{\bf{R}}_m}[n],{{\bf{q}}_k}[n]) ) $ as the achievable rate if UAV $k$ is associated to GBS $m$ in slot $n$. Then the achievable sum rate of the $K$ authorized UAVs in slot $n$ is
\begin{align}
&R^{(\imath ,\jmath)}({{\bf{w}}_{m,k}}[n], {{\bf{R}}_m}[n], {{\bf{q}}_k}[n],{{\alpha} _{m,k}}[n]) \nonumber \\
&= \sum\limits_{m \in {\cal M}}\sum\limits_{k \in {\cal K}} {\alpha _{m,k}}[n]{r_{m,k}^{(\imath ,\jmath)}}({{\bf{W}}_{m,k}}[n], {{\bf{R}}_m}[n],{{\bf{q}}_k}[n]), \nonumber\\
&~~~~ \imath\in \{x,z\}, \jmath \in \{{\text{I}},{\text{II}}\} . \label{sumrate:case1}
\end{align}

\subsection{Sensing Model}

Next, we consider radar sensing towards the targeted 3D area. For facilitating the ISAC design, we focus on  $Q$ sampled locations within the corresponding region as representative sensing points, each of which has an altitude of $H_q$ and a horizontal location of ${\bf v}_q$, $q\in {\cal Q} \buildrel \Delta \over = \{ 1, \ldots ,Q\}$. Based on \eqref{aod:x} and \eqref{aod:z}, we obtain the AoDs from GBS $m$ toward sensing point $q$ as $\theta^{(x)}_m ({{\bf{v}}_q}[n], H_k) $ and $\theta^{(z)}_m ({{\bf{v}}_q}[n], H_k)$ with horizontally and vertically placed ULA antennas, respectively. Accordingly, we obtain the sensing steering vectors as ${{{\bf{a}}}_{m}^{(\imath)}} ({{\bf{v}}_q}[n], H_k)$ defined in \eqref{steer:1}, where $\imath \in \{x,z\}$.

We use the illumination (or received) signal power at the interested sensing locations as the sensing performance metric\footnote{Notice that in practice, enhancing the illumination signal power at given sensing locations generally leads to improved sensing performance in terms of, e.g., detection probability, sensing signal-to-noise ratio (SNR), and target parameters estimation error \cite{Lyu2022Joint, hua2023optimal}. Therefore, the illumination signal power is a proper sensing performance indicator for networked ISAC.}. For target sensing location $q$, the illumination power or the received power from the $M$ GBSs at slot $n$ with antenna configuration $\imath\in\{x,z\}$ is
\begin{align}
 &{\zeta _q^{(\imath)}} ({{\bf{w}}_{l,i}}[n],{{\bf{R}}_l}[n]) = \sum\limits_{l \in {\cal M}}   {{\bf{a}}_l^{(\imath)H}}({{\bf{v}}_q}[n]) (\sum\limits_{i \in {\cal K}} {{{\bf{W}}_{l,i}}[n]} \nonumber \\
&~~~~~~~~ + {{\bf{R}}_l}[n]) {{\bf{a}}^{(\imath)}_l}({{\bf{v}}_q}[n]) /{d_{l,q}}^2.
\end{align}

\subsection{Problem Formulation}
We aim to maximize the average sum rate of $K$ authorized UAVs over the $N$ time slots, by jointly optimizing the coordinated transmit beamforming $\{ {\bf{W}}_{l,i}[n]\}$ and $\{ {\bf{R}}_{l}[n]\}$ at GBSs, the trajectory design $\{{\bf{q}}_k[n]\}$ of authorized UAVs, and the UAV-GBS association $\{\alpha _{m,k}[n]\}$, subject to the minimum illumination power constraints towards the targeted 3D area, the UAV flight constraints, and the maximum transmit power constraints at GBSs. In particular, we assume that the initial and final locations of each UAV $k,\in {\cal K}$ are fixed to be ${\bf{q}}_k[1] = {\bf{q}}_k^{{\rm{I}}}$ and ${\bf{q}}_k[N] = {\bf{q}}_k^{{\rm{F}}}, \forall k \in {\cal K}$, respectively. Also, the UAVs' flights are subject to their individual maximum speed constraints and the collision avoidance constraints, i.e.,
\begin{align}
\left\| {{{\bf{q}}_k}[n + 1] - {{\bf{q}}_k}[n]} \right\| \le {V_{\max }}{\Delta _t}, \forall k \in {\cal K}, n \in {\cal N},
\end{align}
\begin{align}
&\left\| {{{\bf{q}}_k}[n] - {{\bf{q}}_i}[n]} \right\|^2 + (H_k -H_i)^2 \ge {D_{\min }^2}, \nonumber \\
&~~~~~~~~~~~~~~~~~~\forall k ,i \in {\cal K}, k \ne i,n \in {\cal N},
\end{align}
where $V_{\max}$ denotes the maximum UAV speed, and ${D_{\min }}$ denotes the minimum distance between any two UAVs for collision avoidance. Furthermore, to ensure the sensing requirements, we suppose that the illumination power towards each targeted sensing location $q$ should be no less than a given threshold $\Gamma$.

Note that depending on the antenna configuration and the type of UAV receivers employed, there are four different cases under our consideration. We represent them as {\textit {Cases 1-4}} for ease of description.
\begin{itemize}
\item {\textit{Case 1}}: Horizontally placed GBS antennas with $\imath = x$ and Type-I UAV receivers with $\jmath = {\text I}$;
\item {\textit{Case 2}}: Horizontally placed GBS antennas with $\imath = x $ and Type-II UAV receivers with $\jmath = {\text{II}}$;
\item {\textit{Case 3}}: Vertically placed GBS antennas with $\imath = z$ and Type-I UAV receivers with $\jmath = {\text I}$;
\item {\textit{Case 4}}: Vertically placed GBS antennas with $\imath = z$ and Type-II UAV receivers with $\jmath = {\text{II}}$.
\end{itemize}

For {\textit{Case 1}}, the joint coordinated transmit beamforming and UAV trajectory optimization problem for the LAE-oriented networked ISAC system is formulated as
\begin{subequations}
\begin{align}
&( {\rm{P1}} ):  \mathop {\max }\limits_{\{ {{\bf{W}}_{m,k}}[n],{{\bf{R}}_m}[n],{{\bf{q}}_k}[n], {\alpha _{m,k}}[n]\} }  \nonumber \\
&\sum\limits_{n \in {\cal N}} \sum\limits_{m \in {\cal M}}\sum\limits_{k \in {\cal K}} {\alpha _{m,k}}[n]{r_{m,k}^{(x,{\text I})}}({{\bf{W}}_{m,k}}[n], {{\bf{R}}_m}[n],{{\bf{q}}_k}[n]) \label{p1:obj} \\
&~~~~{\rm{s.t.}}  \sum\limits_{l \in {\cal M}} {\bf{a}}_l^{(x )H}({{\bf{v}}_q}[n])(\sum\limits_{i \in {\cal K}} {{{\bf{W}}_{l,i}}[n]}  + {{\bf{R}}_l}[n]) \nonumber \\
&~~~~~~ \cdot {\bf{a}}_l^{(x )} ({{\bf{v}}_q}[n])/{d_{l,q}}^2   \ge \Gamma  , \forall q, n, \label{P2:con:sen} \\
&~~~~~~\sum\limits_{i \in {\cal K}} {\rm{tr}}\left( {{\bf{W}}_{m,i}}[n] \right)  + {\rm{tr}}\left( {{\bf{R}}_m}[n] \right) \le {P_{\max }},\forall m, n, \label{P2:con:pow} \\
&~~~~~~{\bf{q}}_k[1] = {\bf{q}}_k^{{\rm{I}}}, \forall k,\label{p1:ini} \\
&~~~~~~{\bf{q}}_k[N] = {\bf{q}}_k^{{\rm{F}}}, \forall k,\label{p1:end} \\
&~~~~~~\left\| {{{\bf{q}}_k}[n + 1] - {{\bf{q}}_k}[n]} \right\|^2 \le ({V_{\max }}{\Delta _t})^2, \forall k, n, \label{p1:speed} \\
&~~~~~~\left\| {{{\bf{q}}_k}[n] - {{\bf{q}}_i}[n]} \right\|^2 + (H_k -H_i)^2 \ge {D_{\min }^2}, \nonumber \\
&~~~~~~\forall k,i,n, k \ne i, \label{p1:crash} \\
&~~~~~~{\alpha}_{m,k}[n] \in \left\{ {0,1} \right\},  \forall m, k , n, \label{p1:asso1} \\
&~~~~~~\sum\limits_{l \in {\cal M}} {{{\alpha}_{l,k}}}[n]= 1, \forall k, n, \label{p1:asso2} \\
&~~~~~~ {\bf{W}}_{m,k}[n]\succeq 0,{{\bf{R}}_m}[n]\succeq 0, \forall m,k,n, \label{P2:con:semi} \\
&~~~~~~{\rm{rank}}\left( {\bf{W}}_{m,k}[n] \right) \le 1, \forall m , k. \label{P2:con:rank1}
\end{align}
\end{subequations}
Similarly, for {\textit {Case 2}}, we formulate the joint coordinated transmit beamforming and UAV trajectory optimization problem as (P2), which can be obtained based on (P1) by replacing ${r_{m,k} ^{(x ,{\rm{I}})}} ({{\bf{w}}_{m,k}}[n], {{\bf{R}}_m}[n],{{\bf{q}}_k}[n])$ as ${r_{m,k} ^{(x ,{\rm{II}})}} ({{\bf{w}}_{m,k}}[n], {{\bf{R}}_m}[n],{{\bf{q}}_k}[n])$. For {\textit{Case 3}} and {\textit{Case 4}}, we have problems (P3) and (P4), which can be obtained based on (P1) and (P2), respectively, by replacing $\theta^{(x)}_m ({{\bf{q}}_k}[n],H_k)$ as $\theta^{(z)}_m ({{\bf{q}}_k}[n],H_k)$.

Note that problems (P1)-(P4) are mixed-integer non-convex problems that are difficult to solve in general. Take (P1) as an example. The sum rate objective function in \eqref{p1:obj} is highly non-concave, which is due to the fact that the variables $\{ {\bf{w}}_{l,i}[n]\}$, $\{ {\bf{R}}_{l}[n]\}$, $\{{\bf{q}}_k[n]\}$, and $\{\alpha _{m,k}[n]\}$ are coupled, and the trajectory variables $\{{\bf{q}}_k[n]\}$ are involved in the channel matrices ${\bf{H}}_{m}^{(\imath)}({{\bf{q}}_k}[n])$. Moreover, the collision avoidance constraints in \eqref{p1:crash} and the binary UAV association constraints in \eqref{p1:asso1} are also non-convex.

In the following, we first focus on solving the particular problem (P1) in Section \ref{Sec:solution} for {\textit{Case 1}}, and then present the solutions to the other three problems of (P2)-(P4) for the other three cases in \ref{sec:other}.

\section{Proposed Solution to Problem (P1) for {\textit{Case 1}}} \label{Sec:solution}

In this section, we propose an efficient algorithm to solve problem (P1) by using AO. In particular, we alternately optimize the UAV association $\{ \alpha_{m,k}[n]\}$, transmit beamforming $\{{{\bf{W}}_{m,k}}[n]\}$, $\{{{\bf{R}}_{m}}[n]\}$ at GBSs, and UAV trajectory $\{{{\bf{q}}_k}[n]\}$ via using the techniques of SCA and SDR.

\subsection{UAV-GBS Association Optimization} \label{subsec:asso}
In this subsection, we optimize the UAV-GBS association $\{{\alpha _{m,k}}[n]\}$ under given trajectory $\{{\bf q}_k[n]\}$ and transmit beamforming $\{{\bf{W}}_{m,k}[n]\}$ and $\{{\bf{R}}_m[n]\}$. As such, the UAV-GBS association optimization problems is expressed as
\begin{align}
&( {\rm{P1.1}} ):\nonumber \\
&\mathop {\max }\limits_{\{ {\alpha _{m,k}}[n]\} } \sum\limits_{n \in {\cal N}} \sum\limits_{m \in {\cal M}}\sum\limits_{k \in {\cal K}} {\alpha _{m,k}}[n]{r_{m,k}^{(x,{\text I})}}({{\bf{W}}_{m,k}}[n], {{\bf{R}}_m}[n],{{\bf{q}}_k}[n]) \nonumber \\
& ~~~~~~~~~~~~{\rm{s.t.}}~ ~\eqref{p1:asso1} ~{\text{and}}~\eqref{p1:asso2}. \nonumber
\end{align}

Although problem (P1.1) is an integer program, it can be optimally solved by exploiting its special structure. In particular, for any UAV $k$ at slot $n$, it follows that the optimal solution of $\{\alpha_{m,k}[n]\}$ is obtained by ${\alpha^ * _{{m^ * },k}}[n] = 1$ and ${\alpha^ * _{m,k}}[n] = 0,\forall m \ne {m^ * }$, where ${m^ * } = \arg \mathop {\max }\limits_{\{ m\in {\cal M}\} } {r^{(x,{\text I})}} ({{\bf{W}}_{m,k}}[n], {{\bf{R}}_m}[n],{{\bf{q}}_k}[n])$. Notice that there can be multiple GBSs achieving the same maximum data rate. In this case, we can choose any arbitrary one as $m^*$ without loss of optimality.

\subsection{Transmit Beamforming Optimization}\label{subsec:bf}
Next, we optimize the transmit beamforming $\{{\bf{W}}_{m,k}[n]\}$ and $\{{\bf{R}}_m[n]\}$ under given UAV trajectory $\{{\bf q}_k[n]\}$ and UAV association $\{{\alpha _{m,k}}[n]\}$, for which the optimization problem is
\begin{align}
&\left( {\rm{P1.2}} \right):\mathop {\max }\limits_{\{ {\bf{W}}_{m,k}[n],{{\bf{R}}_m}[n]\} } \nonumber \\
&\sum\limits_{n \in {\cal N}} \sum\limits_{m \in {\cal M}}\sum\limits_{k \in {\cal K}} {\alpha _{m,k}}[n]{r_{m,k}^{(x,{\text I})}}({{\bf{W}}_{m,k}}[n], {{\bf{R}}_m}[n],{{\bf{q}}_k}[n]) \nonumber \\
&~~~~~~~~{\rm{s.t.}}~\eqref{P2:con:sen},~\eqref{P2:con:pow}, ~\eqref{P2:con:semi}, ~{\text{and}}~\eqref{P2:con:rank1}. \nonumber
\end{align}

Note that problem (P1.2) is non-convex due to the non-concavity of ${r_{m,k}^{(x,{\text I})}}({{\bf{W}}_{m,k}}[n], {{\bf{R}}_m}[n],{{\bf{q}}_k}[n])$. To address this issue, we adopt the SCA technique to approximate the non-convex objective functions of (P1.2) as a series of concave ones. Consider each SCA iteration $o\ge 1$, in which the local point is denoted by $\{{\bf{W}}_{l,i}^{(o)}[n]\}$ and $\{{\bf{R}}_l^{(o)}[n]\}$. In this case, we first approximate ${r^{(x,{\text I})}}({{\bf{W}}_{m,k}}[n], {{\bf{R}}_m}[n],{{\bf{q}}_k}[n])$ as its lower bound, i.e.,
\begin{align} \label{rate:SCA}
&{r_{m,k}^{(x,{\text I})}}({{\bf{W}}_{m,k}}[n], {{\bf{R}}_m}[n],{{\bf{q}}_k}[n]) \nonumber\\
&\ge {\log _2}  ( {\sum\limits_{l \in {\cal M}} \sum\limits_{i \in {\cal K}} {\rm{tr}}( {\bf{H}}_{m}^{(x)}({{\bf{q}}_k}[n],H_k){\bf{W}}_{l,i}[n] )}  \nonumber\\
&  {+ \sum\limits_{l \in {\cal M}} {\rm{tr}}({\bf{H}}_{m}^{(x)}({{\bf{q}}_k}[n],H_k) {\bf{R}}_l[n])  + \sigma ^2} )- a_{m,k}^{(x,{\text{I}},o)}[n] \nonumber \\
&- \sum\limits_{(l,i) \ne (m,k)} {\rm{tr}}({\bf{B}}_{m}^{(x,{\text{I}},o)} ({{\bf{q}}_k}[n],H_k) \cdot ({\bf{W}}_{l,i}[n] - {\bf{W}}_{l,i}^{(o)}[n])) \nonumber \\
&- \sum\limits_{l\in {\cal M}} {\rm{tr}}({\bf{B}}_{m}^{(x,{\text{I}},o)} ({{\bf{q}}_k}[n],H_k) \cdot ({\bf{R}}_l[n] - {\bf{R}}_l^{(o)}[n]))\nonumber \\
& \triangleq {\bar r}_{m,k}^{(x,{\text I},o)}({{\bf{W}}_{m,k}}[n], {{\bf{R}}_m}[n],{{\bf{q}}_k}[n]),
\end{align}
where ${\bf{B}}_{m}^{(x,{\text{I}},o)} ({{\bf{q}}_k}[n],H_k)$ is defined in \eqref{deri:SCA} at the top of the next page and
\begin{figure*}
\begin{align} \label{deri:SCA}
{\bf{B}}_{m}^{(x,{\text{I}},o)}({{\bf{q}}_k}[n],H_k) = \frac{{\log }_2e \cdot {\bf{H}}_{m}^{(x)}({{\bf{q}}_k}[n],H_k)} {{\sum\limits_{(l,i) \ne (m,k)}  {{\rm{tr}}( {\bf{H}}_{m}^{(x)}({{\bf{q}}_k}[n],H_k) {\bf{W}}_{l,i}^{(o)}[n] )}   + \sum\limits_{l \in {\cal M}} {\rm{tr}}( {\bf{H}}_{m}^{(x)}({{\bf{q}}_k}[n],H_k){\bf{R}}_l^{(o)}[n] )  + \sigma ^2}}.
\end{align}
\vspace{-3ex}
\end{figure*}
\begin{align}\label{deri:SCAa}
a_{m,k}^{(x,{\text{I}},o)}[n]&= {\log _2} ( {\sum\limits_{(l,i) \ne (m,k)}  {\rm{tr}}({\bf{H}}_{m}^{(x)}({{\bf{q}}_k}[n],H_k) {\bf{W}}_{l,i}^{(o)}[n])} \nonumber \\
&  {+ \sum\limits_{l \in {\cal M}} {\rm{tr}}({\bf{H}}_{m}^{(x)}({{\bf{q}}_k}[n],H_k) {\bf{R}}_l^{(o)}[n]) + \sigma ^2} ).
\end{align}

As a result, we substitute ${ r}_{m,k}^{(x,{\text I})}({{\bf{W}}_{m,k}}[n], {{\bf{R}}_m}[n],{{\bf{q}}_k}[n])$ as ${\bar r}_{m,k}^{(x,{\text I},o)}({{\bf{W}}_{m,k}}[n], {{\bf{R}}_m}[n],{{\bf{q}}_k}[n])$ in the objective functions of (P1.2) and thus obtain the approximated transmit beamforming optimization problem as (P1.3.$o$) in the $o$-th SCA iteration, i.e.,
\begin{align}
&( {\rm{P1.3.}}o):\mathop {\max }\limits_{\{ {\bf{W}}_{m,k}[n] ,{{\bf{R}}_m}[n]\} } \nonumber \\
& \sum\limits_{n \in {\cal N}} \sum\limits_{m \in {\cal M}}\sum\limits_{k \in {\cal K}} {\alpha _{m,k}}[n]{ \bar r_{m,k}^{(x,{\text{I}},o)}}({{\bf{W}}_{m,k}}[n], {{\bf{R}}_m}[n],{{\bf{q}}_k}[n]) \nonumber \\
&~~~~~~~~{\rm{s.t.}}~\eqref{P2:con:sen}, ~\eqref{P2:con:pow}, ~\eqref{P2:con:semi},  ~{\text{and}}~\eqref{P2:con:rank1}. \nonumber
\end{align}

However, the rank-one constraints in  \eqref{P2:con:rank1} make problem (P1.3.$o$) still non-convex. To address this issue, we omit the rank-one constraints and obtain the SDR version of (P1.3.$o$) as (SDR1.3.$o$), which are convex and can be optimally solved by standard convex optimization solvers such as CVX \cite{grant2014cvx}. Let $  \{ {\bf{W}}_{m,k}^{ * }[n] \}$ and $\{ {{\bf{R}}_m^ * }[n] \} $  denote the optimal solution to (SDR1.3.$o$). Notice that the obtained solutions of $\{ {{\bf{W}}_{m,k}^ * }[n] \}$ are normally with high ranks and thus may not satisfy the rank-one constraints in \eqref{P2:con:rank1}. We then provide the following propositions to reconstruct equivalent rank-one solutions to problem (P1.3.$o$).

\begin{prop}
If the optimal solutions of $ \{ {\bf{W}}_{l,i}^ {*}[n] \} $ to problem (SDR1.3.$o$) are not rank-one, then we reconstruct the equivalent solutions to  problem (P1.3.$o$) as $ \{ {{\bf{\bar W}}}_{l,i}[n]\}$ and $\{ {\bf{\bar R}}_l[n]\} $, given by
\begin{subequations}\label{app:1}
\begin{align}
{\bf{\bar w}}_{l,i}[n] &= ( {{\bf{h}}_{m}^{(x)H}({{\bf{q}}_k}[n],H_k){\bf{W}}_{l,i}^{*}[n] {{\bf{h}}_{m}^{(x)}({{\bf{q}}_k}[n],H_k)}} )^{ -  \frac{1}{2}}\nonumber\\
&~~~~\cdot{\bf{W}}_{l,i}^{*}[n]{{\bf{h}}_{m}^{(x)} ({{\bf{q}}_k}[n],H_k)}, \label{app:1:1}  \\
{{{\bf{\bar W}}}_{l,i}[n]} &= {{\bf{\bar w}}_{l,i}[n]}{\bf{\bar w}}_{l,i}^H[n], \label{app:1:2} \\
{\bf{\bar R}}_l[n] &= \sum\limits_{i \in {\cal K}} {{\bf{W}}_{l,i}^{*}}[n]  + {\bf{R}}_l^{*}[n] - \sum\limits_{i \in {\cal K}} {\bf{\bar W}}_{l,i}[n], \label{app:1:3}
\end{align}
\end{subequations}
which satisfy the rank-one constraints and are feasible for problem (P1.3.$o$). The equivalent solutions in \eqref{app:1} achieve the same objective value for (P1.3.$o$) as the optimal value achieved by $ \{ {\bf{W}}_{m,k}^{ * }[n] \}$ and $\{ {{\bf{R}}_m^ * }[n] \} $ for problem (P1.3.$o$). Therefore, $\{{\bf{\bar w}}_{l,i}[n] \}$ and $\{{\bf{\bar R}}_l[n]\}$ are optimal for (P1.3.$o$) and thus the SDR of (P1.3.$o$) is tight.
\end{prop}

\begin{proof}
See Appendix \ref{appendixA}.
\end{proof}



Therefore, in each SCA iteration $o$, we obtain the optimal solution to (P1.3.$o$), which can be shown to lead to monotonically non-decreasing objective values for (P1.2). Hence, the SCA-based solution to problem (P1.2) is ensured to converge.

\subsection{UAV Trajectory Optimization} \label{subsec:tra}
Next, we optimize the UAV trajectory $\{{\bf q}_k[n]\}$ with given transmit beamforming $\{{\bf{W}}_{l,i}[n]\}$, $\{{{\bf{R}}_l}[n]\}$ and UAV-GBS association $\{{ \alpha}_{l,i}[n]\}$. The corresponding optimization problem becomes
\begin{align}
&( {\rm{P1.4}} ):\nonumber \\
&\mathop {\max }\limits_{\{ {\bf q}_k[n]\} }\sum\limits_{n \in {\cal N}} \sum\limits_{m \in {\cal M}}\sum\limits_{k \in {\cal K}} {\alpha _{m,k}}[n]{r_{m,k}^{(x,{\text I})}}({{\bf{W}}_{m,k}}[n], {{\bf{R}}_m}[n],{{\bf{q}}_k}[n])    \nonumber \\
&~~~~~~~~{\rm{s.t.}}~\eqref{p1:ini}, ~\eqref{p1:end}, ~\eqref{p1:speed}, ~{\text{and}}~\eqref{p1:crash}. \nonumber
\end{align}
Note that in (P1.4), \eqref{p1:crash} are non-convex constraints and the optimization variables $\{{\bf q}_k[n]\}$ appear in the AoD, both of which are difficult to be tackled.

First, we adopt SCA to approximate \eqref{p1:crash}. In each iteration $o$, the first-order Taylor expansion is applied to approximate the left-hand-side of \eqref{p1:crash} with respect to ${\bf{q}}_k^{(o)}[n]$ and ${\bf{q}}_i^{(o)}[n]$, based on which the approximate version of constraint \eqref{p1:crash} is expressed as
\begin{align}
&  2( {\bf{q}}_k^{(o)}[n] - {\bf{q}}_i^{(o)}[n] )^T \left( {\bf{q}}_k[n] - {\bf{q}}_i[n] \right) \nonumber \\
&- || {\bf{q}}_k^{(o)}[n] - {\bf{q}}_i^{(o)}[n] || ^2\ge D_{\min }^2 - (H_k-H_i)^2. \label{SCA:crash}
\end{align}

Next, we reformulate ${r_{m,k}^{(x,{\text I})}} ({{\bf{W}}_{m,k}}[n], {{\bf{R}}_m}[n], {{\bf{q}}_k}[n])$ in the objective function. For ease of derivation, we denote the entries in the $p$-th and $q$-th column of ${\bf{W}}_{l,i}[n]$ and ${\bf{R}}_{l}[n]$ as $[ {\bf{W}}_{l,i}[n] ]_{p,q}$ and $\left[ {\bf{R}}_{l}[n] \right]_{p,q}$, respectively. Similarly, their absolute values are denoted by $| {\left[ {\bf{W}}_{l,i}[n] \right]}_{p,q} |$ and $| {\left[ {\bf{R}}_{l}[n] \right]}_{p,q} |$. Moreover, the phases of these entries are denoted by $\theta _{p,q}^{{\bf{W}}_{l,i}[n]}$ and $\theta _{p,q}^{{\bf{R}}_{l}[n]}$, respectively. It thus follows that
\begin{align}
&{r_{m,k}^{(x,{\text I})}}({{\bf{W}}_{m,k}}[n], {{\bf{R}}_m}[n],{{\bf{q}}_k}[n]) \nonumber \\
&= {\log _2} ( {\sum\limits_{l \in {\cal M}} \sum\limits_{i \in {\cal K}}  {\eta^{(x)}  \left( {\bf{W}}_{l,i}[n],{\bf{q}}_k[n] \right)} } \nonumber \\
& +  \sum\limits_{l \in {\cal M}}\mu^{(x)} \left({\bf{R}}_l[n],{\bf{q}}_k[n]\right)  + \frac{{{\sigma ^2}}}{\kappa }({\left\| {{{\bf{q}}_k}[n] - {{\bf{u}}_m}} \right\|^2} + {H_k^2}) )  \nonumber \\
&- {\log _2} ( \sum\limits_{(l,i) \ne (m,k)} \eta^{(x)} \left( {\bf{W}}_{l,i}[n],{\bf{q}}_k[n] \right)\nonumber \\
& + \sum\limits_{l \in {\cal M}} \mu^{(x)} \left({\bf{R}}_l[n],{\bf{q}}_k[n]\right)  + \frac{{{\sigma ^2}}}{\kappa }({\left\| {{{\bf{q}}_k}[n] - {{\bf{u}}_m}} \right\|^2} + {H_k^2})), \label{traj:rate}
\end{align}
where $\eta^{(x)} \left({\bf{W}}_{l,i}[n],{\bf{q}}_k[n]\right)$ and  $\mu^{(x)} ({\bf{R}}_l[n],{\bf{q}}_k[n])$ are expressed as
\begin{align}
&\eta^{(x)} \left({\bf{W}}_{l,i}[n],{\bf{q}}_k[n]\right) \nonumber \\
&=  \sum\limits_{r = 1}^{{N_a}} {[ {\bf{W}}_{l,i}[n] ]}_{r,r}  +  2\sum\limits_{p = 1}^{N_a}   \sum\limits_{q = p+1}^{{N_a}}  | {\left[ {\bf{W}}_{l,i}[n] \right]}_{p,q} | \nonumber \\
&\times \cos (\theta _{p,q}^{{{\bf{W}}_{l,i}}[n]} + 2\pi \frac{d}{\lambda }(q - p)\frac{{{{\boldsymbol{\psi }}^T}({{\bf{q}}_k}[n] - {{\bf{u}}_m})}}{\sqrt {{{\left\| {{{\bf{q}}_k}[n] - {{\bf{u}}_m}} \right\|}^2} + H_k^2} }) ,
\end{align}
and
\begin{align}
&\mu^{(x)} ({\bf{R}}_l[n],{\bf{q}}_k[n])\nonumber \\
&=  \sum\limits_{r = 1}^{{N_a}} {\left[ {\bf{R}}_l [n] \right]}_{r,r} + 2\sum\limits_{p = 1}^{N_a} \sum\limits_{q = p+1}^{N_a} | {\left[ {\bf{R}}_l[n] \right]}_{p,q} | \nonumber \\
&\times \cos ( \theta _{p,q}^{{\bf{R}}_{l}[n]} + 2\pi \frac{d}{\lambda }(q - p)\frac{{\boldsymbol{\psi }}^T({{\bf{q}}_k}[n] - {{\bf{u}}_m})}{{\sqrt {{{\left\| {{{\bf{q}}_k}[n] - {{\bf{u}}_m}} \right\|}^2} + H_k^2} } } ).
\end{align}

In the following, we adopt the first-order Taylor expansion to approximate \eqref{traj:rate} in iteration $o$ as
\begin{align}
&{r^{(x,{\text I})}}({{\bf{W}}_{m,k}}[n], {{\bf{R}}_m}[n],{{\bf{q}}_k}[n])\nonumber \\
& \approx c_{m,k}^{(x,{\text{I}},o)}[n] + {\bf{d}}_{m,k}^{(x,{\text{I}},o)T}[n]({\bf{q}}_k[n] - {\bf{q}}_k^{(o)}[n])\nonumber \\
&\triangleq  {\tilde r^{(x,{\text I},o)}}({{\bf{W}}_{m,k}}[n], {{\bf{R}}_m}[n],{{\bf{q}}_k}[n]) , \label{QoS:linear}
\end{align}
where $c_{m,k}^{(x,{\text{I}},o)}[n] $ and ${\bf{d}}_{m,k}^{(x,{\text{I}},o)}[n]$ are defined in \eqref{SCA:cmk} and \eqref{SCA:dmk} with ${g_{m,k}^{(x,{\text{I}},o)}}[n]$, ${h_{m,k}^{(x,{\text{I}},o)}}[n]$, $\nu^{(x)} ({\bf{W}}_{l,i}[n],{\bf{q}}_k^{(o)}[n])$, and $\upsilon^{(x)} ({\bf{R}}_l[n],{\bf{q}}_k^{(o)}[n])$  expressed in \eqref{SCA:gmk}, \eqref{SCA:hmk}, \eqref{SCA:nu}, and \eqref{SCA:upsilon}, respectively.
\begin{align}\label{SCA:cmk}
&c_{m,k}^{(x,{\text{I}},o)}[n]={\log _2} ( \sum\limits_{l \in {\cal M}} \sum\limits_{i \in {\cal K}} {\eta^{(x)} ( {\bf{W}}_{l,i}[n] , {\bf{q}}_k^{(o)}[n] )} \nonumber \\
&  + \sum\limits_{l \in {\cal M}}  \mu^{(x)} ({\bf{R}}_l[n]  , {\bf{q}}_k^{(o)}[n])+ \frac{{{\sigma ^2}}}{\kappa } ({|| {{\bf{q}}_k^{(o)}[n] - {{\bf{u}}_m}} ||^2} + H_k^2 ))\nonumber \\
& - {\log _2} ( \sum\limits_{(l,i) \ne (m,k)} {\eta^{(x)} ({{\bf{W}}_{l,i}}[n],{\bf{q}}_k^{(o)}[n])}  \nonumber \\
&+ \sum\limits_{l \in {\cal M}}\mu^{(x)} ({\bf{R}}_l[n],{\bf{q}}_k^{(o)}[n]) +  \frac{{{\sigma ^2}}}{\kappa }({ || {{{\bf{q}}_k^{(o)}}[n] - {{\bf{u}}_m}}|| ^2} + {H_k^2})),
\end{align}
\begin{align}\label{SCA:dmk}
&{\bf{d}}_{m,k}^{(x,{\text{I}},o)}[n]= \frac{{ {{\log }_2}e}}{g_{m,k}^{(x,{\text{I}},o)}[n]} (  \sum\limits_{l \in {\cal M}} {\sum\limits_{i \in {\cal K}} \nu^{(x)} ({\bf{W}}_{l,i}[n],{\bf{q}}_k^{(o)}[n]) } \nonumber \\
&  + \sum\limits_{l \in {\cal M}} \upsilon^{(x)} ({{\bf{R}}_l}[n],{\bf{q}}_k^{(o)}[n])  + \frac{{2{\sigma ^2}}}{\kappa }({{\bf{q}}_k}[n] - {{\bf{u}}_m}) )\nonumber \\
& - \frac{ {{\log }_2}e}{h_{m,k} ^{(x,{\text{I}},o)}[n]}  ( \sum\limits_{(l,i) \ne (m,k)} {\nu^{(x)} ({{\bf{W}}_{l,i}}[n],{\bf{q}}_k^{(o)}[n])} \nonumber \\
& + \sum\limits_{l \in {\cal M}} \upsilon^{(x)} ({{\bf{R}}_l}[n],{\bf{q}}_k^{(o)}[n])  + \frac{{2{\sigma ^2}}}{\kappa }({{\bf{q}}_k}[n] - {{\bf{u}}_m}) ) ,
\end{align}
\begin{align}\label{SCA:gmk}
&{g_{m,k}^{(x,{\text{I}},o)}}[n]= \sum\limits_{l \in {\cal M}} \sum\limits_{i \in {\cal K}} \eta^{(x)}( {\bf{W}}_{l,i}[n],{\bf{q}}_k^{(o)}[n] )\nonumber \\
&+  \sum\limits_{l \in {\cal M}} \mu^{(x)} ({\bf{R}}_l[n],{\bf{q}}_k^{(o)}[n])+ \frac{{{\sigma ^2}}}{\kappa }({ || {{{\bf{q}}_k^{(o)}}[n] - {{\bf{u}}_m}} || ^2} + {H_k^2}),
\end{align}
\begin{align}\label{SCA:hmk}
&{h_{m,k}^{(x,{\text{I}},o)}}[n]= \sum\limits_{(l,i) \ne (m,k)} \eta^{(x)} ( {\bf{W}}_{l,i}[n],{\bf{q}}_k^{(o)}[n] )\nonumber \\
&+ \sum\limits_{l \in {\cal M}}\mu^{(x)} ({{\bf{R}}_l[n]},{\bf{q}}_k^{(o)}[n]) + \frac{{{\sigma ^2}}}{\kappa }({\left\| {{{\bf{q}}_k}[n] - {{\bf{u}}_m}} \right\|^2} + {H_k^2}).
\end{align}
\begin{figure*}
\begin{align}\label{SCA:nu}
\nu^{(x)} ({{\bf{W}}_{l,i}}[n],{\bf{q}}_k^{(o)}[n]) &= \sum\limits_{p = 1}^{{N_a}} \sum\limits_{q = p + 1}^{{N_a}}  - 4\pi \frac{d}{\lambda }(q - p)| {{{[{{\bf{W}}_{l,i}}[n]]}_{p,q}}} | \sin ( {\theta _{p,q}^{{{\bf{W}}_{l,i}}[n]} + 2\pi \frac{d}{\lambda }(q - p)\frac{{{{\boldsymbol{\psi }}^T}({\bf{q}}_k^{(o)}[n] - {{\bf{u}}_m})}}{{\sqrt {{{|| {{\bf{q}}_k^{(o)}[n] - {{\bf{u}}_m}} ||}^2} + H_k^2} }}} ) \nonumber \\
&\cdot ( {\frac{{\boldsymbol{\psi }}}{{\sqrt {{{|| {{\bf{q}}_k^{(o)}[n] - {{\bf{u}}_m}} ||}^2} + H_k^2} }} - \frac{{{{\boldsymbol{\psi }}^T}({\bf{q}}_k^{(o)}[n] - {{\bf{u}}_m}) \cdot ({\bf{q}}_k^{(o)}[n] - {{\bf{u}}_m})}}{{{{\sqrt {{{|| {{\bf{q}}_k^{(o)}[n] - {{\bf{u}}_m}} ||}^2} + H_k^2} }^3}}}} )   .
\end{align}
\end{figure*}
\begin{figure*}
\begin{align}\label{SCA:upsilon}
\upsilon^{(x)} ({{\bf{R}}_l}[n],{\bf{q}}_k^{(o)}[n]) &= \sum\limits_{p = 1}^{{N_a}} \sum\limits_{q = p + 1}^{{N_a}}  - 4\pi \frac{d}{\lambda }(q - p)| {{{[{{\bf{R}}_l}[n]]}_{p,q}}} |\sin ( {\theta _{p,q}^{{{\bf{R}}_l}[n]} + 2\pi \frac{d}{\lambda }(q - p)\frac{{{{\boldsymbol{\psi }}^T}({\bf{q}}_k^{(o)}[n] - {{\bf{u}}_m})}}{{\sqrt {{{|| {{\bf{q}}_k^{(o)}[n] - {{\bf{u}}_m}} ||}^2} + H_k^2} }}} )\nonumber \\
&\cdot ( {\frac{{\boldsymbol{\psi }}}{{\sqrt {{{|| {{\bf{q}}_k^{(o)}[n] - {{\bf{u}}_m}} ||}^2} + H_k^2} }} - \frac{{{{\boldsymbol{\psi }}^T}({\bf{q}}_k^{(o)}[n] - {{\bf{u}}_m}) \cdot ({\bf{q}}_k^{(o)}[n] - {{\bf{u}}_m})}}{{{{\sqrt {{{|| {{\bf{q}}_k^{(o)}[n] - {{\bf{u}}_m}} ||}^2} + H_k^2} }^3}}}} )  .
\end{align}
\vspace{-3ex}
\end{figure*}

After approximating the non-convex objective function and the collision avoidance constraints \eqref{SCA:crash} for problem (P1.4), we further introduce additional constraints to ensure the accuracy of our approximation. In particular, we consider a series of trust region constraints in each iteration $o$ \cite{conn2000trust}, i.e.,
\begin{align}
|| {{\bf{q}}_k^{( o)}[n] - {\bf{q}}_k^{( o - 1)}[n]}||  \le {\omega ^{( o)}}, \forall k \in {\cal K}, n \in {\cal N}, \label{trust:region}
\end{align}
where ${\omega ^{( o)}}$ denotes the trust region radius. Thus, we obtain the approximated version of UAV trajectory optimization problems as (P1.5.$o$) in each SCA iteration $o$.
\begin{align}
&\left( {\rm{P1.5}}.o \right):\nonumber \\
&\mathop {\max }\limits_{\{ {\bf q}_k[n]\} } \!\! \sum\limits_{n \in {\cal N}} \sum\limits_{m \in {\cal M}}\sum\limits_{k \in {\cal K}} {\alpha _{m,k}}[n] {\tilde r_{m,k}^{(x,{\text I},o)}}({{\bf{W}}_{m,k}}[n], {{\bf{R}}_m}[n],{{\bf{q}}_k}[n])  \nonumber \\
&~{\rm{s.t.}}~\eqref{p1:ini}, ~\eqref{p1:end}, ~\eqref{p1:speed}, ~\eqref{SCA:crash},~{\text{and}}~\eqref{trust:region}.
\end{align}

In general, we solve problem (P1.4) by iteratively solving a series of convex problems (P1.5.$o$)'s. The convergence of the SCA-based algorithm with trust region can be ensured by choosing  ${\omega ^{( o)}}$ \cite{conn2000trust} to be sufficiently small. In practice, in each iteration $o$, we reduce the trust region radius by ${\omega ^{(o)}} = \frac{1}{2}{\omega ^{(o)}}$ if the objective value is non-decreasing. The reduction will be terminated when ${\omega ^{(o)}}$ is lower than a threshold $\xi$.

Therefore, the algorithms for solving subproblems (P1.1), (P1.2), and (P1.4) are convergent by obtaining the optimal solution or solving a series of SCA problems. As a result, the whole AO-based algorithm is is ensured to converge.

\section{Proposed Solution to Problems (P2)-(P4)}\label{sec:other}
In this section, we extend our proposed solution of (P1) to solve problems (P2)-(P4) for the other three cases.

\subsection{Solution to Problems (P2) for {\textit{Case 2}}}\label{subsec:p2}
In {\textit{Case 2}}, we consider horizontally placed antennas at GBSs and Type-II UAV receivers. We also adopt AO to solve problem (P2) by iteratively optimizing the UAV-GBS association, transmit beamforming, and UAV trajectory. Notice that the UAV-GBS association optimization problem can be similarly solved as for problem (P1.1). In the following, we only need to focus on the transmit beamforming optimization and UAV trajectory optimization, respectively.

First, we use the SCA to tackle the transmit beamforming optimization problem:
\begin{align}
&\left( {\rm{P2.1}} \right):\mathop {\max }\limits_{\{ {\bf{W}}_{m,k}[n],{{\bf{R}}_m}[n]\} } \nonumber \\
&\sum\limits_{n \in {\cal N}} \sum\limits_{m \in {\cal M}}\sum\limits_{k \in {\cal K}} {\alpha _{m,k}}[n]{r_{m,k}^{(x,{\text {II}})}}({{\bf{W}}_{m,k}}[n], {{\bf{R}}_m}[n],{{\bf{q}}_k}[n]) \nonumber \\
&~~~~~~~~{\rm{s.t.}}~\eqref{P2:con:sen},~\eqref{P2:con:pow}, ~\eqref{P2:con:semi}, ~{\text{and}}~\eqref{P2:con:rank1}. \nonumber
\end{align}

Towards this end, in each SCA iteration o, we approximate ${r_{m,k}^{(x,{\text {II}})}}({{\bf{W}}_{m,k}}[n], {{\bf{R}}_m}[n],{{\bf{q}}_k}[n])$ as
\begin{align} \label{rate:SCAII}
&{r_{m,k}^{(x,{\text {II}})}}({{\bf{W}}_{m,k}}[n], {{\bf{R}}_m}[n],{{\bf{q}}_k}[n]) \nonumber \\
&\ge {\log _2}  ( {\sum\limits_{l \in {\cal M}} \sum\limits_{i \in {\cal K}} {\rm{tr}}( {\bf{H}}_{m}^{(x)}({{\bf{q}}_k}[n],H_k) {\bf{W}}_{l,i}[n] ) + \sigma ^2 } ) \nonumber\\
&- a_{m,k}^{(x,{\text{II}},o)}[n] - \sum\limits_{(l,i) \ne (m,k)} {\rm{tr}}({\bf{B}}_{m}^{(x,{\text{II}},o)} ({{\bf{q}}_k}[n],H_k)\nonumber \\
&\cdot ({\bf{W}}_{l,i}[n] -  {\bf{W}}_{l,i}^{(o)}[n])) \nonumber \\
& \triangleq {\bar r_{m,k}^{(x,{\text {II}},o)}}({{\bf{W}}_{m,k}}[n], {{\bf{R}}_m}[n],{{\bf{q}}_k}[n]),
\end{align}
where ${\bf{B}}_{m}^{(x,{\text{II}},o)} ({{\bf{q}}_k}[n],H_k)$ and $a_{m,k}^{(x,{\text{II}},o)}[n]$ are defined as
\begin{align} \label{deri:SCAII}
&{\bf{B}}_{m,k}^{(x,{\text{II}},o)} ({{\bf{q}}_k}[n],H_k) \nonumber\\
&= \frac{{\log }_2e\cdot{\bf{H}}_{m}^{(x)} ({{\bf{q}}_k}[n],H_k)}{\sum\limits_{(l,i) \ne (m,k)}  {{\rm{tr}}( {\bf{H}}_{m}^{(x)}({{\bf{q}}_k}[n],H_k) {\bf{W}}_{l,i}^{(o)}[n] )} + \sigma ^2},
\end{align}
and
\begin{align}\label{deri:SCAaII}
&a_{m,k}^{(x,{\text{II}},o)}[n]\nonumber \\
&= {\log _2} ( {\sum\limits_{(l,i) \ne (m,k)}  {\rm{tr}}({\bf{H}}_{m}^{(x)}({{\bf{q}}_k}[n],H_k){\bf{W}}_{l,i}^{(o)}[n])+ \sigma ^2} ).
\end{align}
By substituting ${r_{m,k}^{(x,{\text {II}})}}({{\bf{W}}_{m,k}}[n], {{\bf{R}}_m}[n],{{\bf{q}}_k}[n])$ as ${\bar r^{(x,{\text {II}})}}({{\bf{W}}_{m,k}}[n], {{\bf{R}}_m}[n],{{\bf{q}}_k}[n])$ in the objective function in (P2.1), we obtain the approximated transmit beamforming optimization problem as (P2.2.$o$) in the $o$-th SCA iteration, i.e.,
\begin{align}
&( {\rm{P2.2.}}o):\mathop {\max }\limits_{\{ {\bf{W}}_{m,k}[n] ,{{\bf{R}}_m}[n]\} } \nonumber \\
& \sum\limits_{n \in {\cal N}} \sum\limits_{m \in {\cal M}}\sum\limits_{k \in {\cal K}} {\alpha _{m,k}}[n]{ \bar r_{m,k}^{(x,{\text {II}},o)}}({{\bf{W}}_{m,k}}[n], {{\bf{R}}_m}[n],{{\bf{q}}_k}[n]) \nonumber \\
&~~~~~~~~{\rm{s.t.}}~\eqref{P2:con:sen}, ~\eqref{P2:con:pow}, ~\eqref{P2:con:semi},  ~{\text{and}}~\eqref{P2:con:rank1}. \nonumber
\end{align}
To solve  (P2.2.$o$), we omit the non-convex rank-one constraints \eqref{P2:con:rank1} to obtain (SDR2.2.$o$). We have the following proposition to find the rank-one solution to (P2.2.$o$) by equivalently solving the convex semidefinite programming (SDP) (SDR2.2.$o$).
\begin{prop}
Suppose that the optimal solution to (SDR2.2.$o$) is given by $ \{ {\bf{W}}_{l,i}^ {**}[n] \} $ and $ \{ {\bf{R}}_{l,i}^ {**}[n] \} $, in which $ \{ {\bf{W}}_{l,i}^ {**}[n] \} $ may not be not rank-one in general. Then we reconstruct the equivalent solutions to  problem (P2.2.$o$) as $ \{ {{\bf{\tilde W}}}_{l,i}[n]\}$ and $\{ {\bf{\tilde  R}}_l[n]\} $, given by
\begin{subequations}\label{app:2}
\begin{align}
{\bf{\tilde  w}}_{l,i}[n] &= ( {{\bf{h}}_{m}^{(x)H}({{\bf{q}}_k}[n],H_k){\bf{W}}_{l,i}^{**}[n] {\bf{h}}_{m}^{(x)}({{\bf{q}}_k}[n],H_k)} )^{ -  \frac{1}{2}} \nonumber \\
&~~~~\cdot {\bf{W}}_{l,i}^{**}[n] {{\bf{h}}_{m}^{(x)}({{\bf{q}}_k}[n])}, \label{app:2:1}  \\
{{{\bf{\tilde  W}}}_{l,i}[n]} &= {{\bf{\tilde  w}}_{l,i}[n]}{\bf{\tilde  w}}_{l,i}^H[n], \label{app:2:2} \\
{\bf{\tilde  R}}_l[n] &= \sum\limits_{i \in {\cal K}} {{\bf{W}}_{l,i}^{**}}[n] + {\bf{R}}_l^{**}[n] - \sum\limits_{i \in {\cal K}} {\bf{\tilde  W}}_{l,i}[n], \label{app:2:3}
\end{align}
\end{subequations}
which satisfy the rank-one constraints and are feasible for problem (P2.2.$o$). The equivalent solutions in \eqref{app:2} achieve the same optimal value for (P2.2.$o$) with $ \{ {\bf{W}}_{m,k}^{ ** }[n] \}$ and $\{ {{\bf{R}}_m^ {**} }[n] \} $ as well as (SDR2.2.$o$). Therefore, $\{{\bf{\tilde w}}_{l,i}[n] \}$ and $\{{\bf{\tilde R}}_l[n]\}$ are optimal for (P2.2.$o$) and the SDR (SDR2.2.$o$) is tight.
\end{prop}
\begin{proof}
The proof is similar to {\textit{Proposition 1}} and thus omitted.
\end{proof}

Next, we handle the trajectory optimization problem via SCA. Towards this end, in each SCA iteration $o$ with local point ${\bf q}_k[n]$, the average sum rate is approximated as
\begin{align}
&{r_{m,k}^{(x,{\text {II}})}}({{\bf{W}}_{m,k}}[n], {{\bf{R}}_m}[n],{{\bf{q}}_k}[n])\nonumber \\
& \approx c_{m,k}^{(x,{\text{II}},o)}[n] + {\bf{d}}_{m,k}^{(x,{\text{II}},o)T}[n]({\bf{q}}_k[n] - {\bf{q}}_k^{(o)}[n])\nonumber \\
&\triangleq  {\tilde r^{(x,{\text {II}},o)}}({{\bf{W}}_{m,k}}[n], {{\bf{R}}_m}[n],{{\bf{q}}_k}[n]) , \label{QoS:linearII}
\end{align}
where $c_{m,k}^{(x,{\text{II}},o)}[n] $ and ${\bf{d}}_{m,k}^{(x,{\text{II}},o)}[n]$ are defined in \eqref{SCA:cmk:II} and \eqref{SCA:dmk:II} with ${g_{m,k}^{(x,{\text{II}},o)}}[n]$ and ${h_{m,k}^{(x,{\text{II}},o)}}[n]$ defined as \eqref{SCA:gmk:II} and \eqref{SCA:hmk:II}
\begin{align}\label{SCA:cmk:II}
c_{m,k}^{(x,{\text{II}},o)}[n]&={\log _2} ( \sum\limits_{l \in {\cal M}} \sum\limits_{i \in {\cal K}} {\eta^{(x)} ( {\bf{W}}_{l,i}[n] , {\bf{q}}_k^{(o)}[n] )} \nonumber \\
&+ \frac{{{\sigma ^2}}}{\kappa } (|| {{\bf{q}}_k^{(o)}[n] - {{\bf{u}}_m}} ||^2 + H_k^2 )) \nonumber \\
&- {\log _2} ( \sum\limits_{(l,i) \ne (m,k)}{\eta^{(x)} ({{\bf{W}}_{l,i}}[n], {\bf{q}}_k^{(o)}[n])} \nonumber \\
& +  \frac{{{\sigma ^2}}}{\kappa }({ || {{{\bf{q}}_k^{(o)}}[n] - {{\bf{u}}_m}}|| ^2} + {H_k^2})),
\end{align}
\begin{align}\label{SCA:dmk:II}
{\bf{d}}_{m,k}^{(x,{\text{II}},o)}[n]&= \frac{{ {{\log }_2}e}}{g_{m,k}^{(x,{\text{II}},o)}[n]} (\sum\limits_{l \in {\cal M}} \sum\limits_{i \in {\cal K}} \nu^{(x)} ({\bf{W}}_{l,i}[n],{\bf{q}}_k^{(o)}[n]) \nonumber \\
& + \frac{{2{\sigma ^2}}}{\kappa }({{\bf{q}}_k}[n] - {{\bf{u}}_m})) \nonumber \\
&- \frac{ {{\log }_2}e}{h_{m,k}^{(x,{\text{II}},o)}[n]} (\sum\limits_{(l,i) \ne (m,k)}\nu^{(x)} ({{\bf{W}}_{l,i}}[n],{\bf{q}}_k^{(o)}[n]) \nonumber \\
&+ \frac{{2{\sigma ^2}}}{\kappa }({{\bf{q}}_k}[n] - {{\bf{u}}_m}) ),
\end{align}
\begin{align}\label{SCA:gmk:II}
{g_{m,k}^{(x,{\text{II}},o)}}[n] &= \sum\limits_{l \in {\cal M}} \sum\limits_{i \in {\cal K}} \eta^{(x)}( {\bf{W}}_{l,i}[n],{\bf{q}}_k^{(o)}[n] ) \nonumber \\
&+ \frac{{{\sigma ^2}}}{\kappa }({ || {{{\bf{q}}_k^{(o)}}[n] - {{\bf{u}}_m}} || ^2} + {H_k^2}),
\end{align}
\begin{align}\label{SCA:hmk:II}
{h_{m,k}^{(x,{\text{II}},o)}}[n] &= \sum\limits_{(l,i) \ne (m,k)} \eta^{(x)} ( {\bf{W}}_{l,i}[n],{\bf{q}}_k^{(o)}[n] ) \nonumber \\
&+ \frac{{{\sigma ^2}}}{\kappa }({\left\| {{{\bf{q}}_k}[n] - {{\bf{u}}_m}} \right\|^2} + {H_k^2}).
\end{align}
Based on the above approximation, we can obtain the approximated UAV trajectory optimization problem in each iteration $o$. By using the trust region method, we can find a converged solution to the UAV trajectory optimization problem, for which the details are similar to those in Section \ref{subsec:tra} and thus omitted. By combining the solutions of the above three subproblems together with AO, problem (P2) can be efficiently solved.

\subsection{Solutions to Problems (P3) and (P4)}
In {\textit{Case 3}} and {\textit{Case 4}}, problems (P3) and (P4) correspond to the scenario with vertically placed antennas GBSs, under Type-I and Type-II UAV receivers, respectively. We also use the AO to solve the two problems similarly as in Section \ref{Sec:solution} and Section \ref{subsec:p2}. In the two cases, the UAV-GBS association and coordinated transmit beamforming design are similar to those for problems (P1) and (P2), and thus the solutions are omitted for brevity. In the following, we only focus on the UAV trajectory design with vertically placed antennas at GBSs.


In general, we can use an algorithm similar to that in Section \ref{subsec:tra} to solve the sub-problem of UAV trajectory optimization in {\textit {Case 3}} and {\textit {Case 4}}. In particular, we only need to derive new approximations for ${\eta^{(z)} \left( {\bf{W}}_{l,i}[n],{\bf{q}}_k[n] \right)}$ and $\mu^{(z)} \left({\bf{R}}_l[n],{\bf{q}}_k[n]\right)$, which are respectively expressed as
\begin{align}
&\eta^{(z)} \left({\bf{W}}_{l,i}[n],{\bf{q}}_k[n]\right) \nonumber \\
&=  \sum\limits_{r = 1}^{{N_a}} {[ {\bf{W}}_{l,i}[n] ]}_{r,r} + 2\sum\limits_{p = 1}^{N_a}   \sum\limits_{q = p+1}^{{N_a}}  | {\left[ {\bf{W}}_{l,i}[n] \right]}_{p,q} | \nonumber \\
&\times \cos ( \theta _{p,q}^{{\bf{W}}_{l,i}[n]}+ 2\pi \frac{d}{\lambda }(q - p)\frac{H_k}{\sqrt{ {{\left\| {{{\bf{q}}_k}[n]  - {{\bf{u}}_m}} \right\|}^2 + {H_k^2}}} } ) ,
\end{align}
and
\begin{align}
&\mu^{(z)} ({\bf{R}}_l[n],{\bf{q}}_k[n])\nonumber \\
&= \sum\limits_{r = 1}^{{N_a}} {\left[ {\bf{R}}_l [n] \right]}_{r,r}  + 2\sum\limits_{p = 1}^{N_a} \sum\limits_{q = p+1}^{N_a} | {\left[ {\bf{R}}_l[n] \right]}_{p,q} | \nonumber \\
&\times \cos ( \theta _{p,q}^{{\bf{R}}_{l}[n]} + 2\pi \frac{d}{\lambda }(q - p)\frac{H_k}{\sqrt{ {{\left\| {{{\bf{q}}_k}[n] - {{\bf{u}}_m}} \right\|}^2 + {H_k^2}} }} ).
\end{align}
Moreover, we derive the new approximation $\nu^{(z)} ({{\bf{W}}_{l,i}}[n],{\bf{q}}_k^{(o)}[n])$ and $\upsilon^{(z)} ({{\bf{R}}_l}[n],{\bf{q}}_k^{(o)}[n])$, which are given in \eqref{SCA:nu:z} and \eqref{SCA:upsilon:z}, respectively, at the top of next page.
\begin{figure*}
\begin{align}\label{SCA:nu:z}
\nu^{(z)} ({\bf{W}}_{l,i}[n],{\bf{q}}_k^{(o)}[n])&= \sum\limits_{p = 1}^{{N_a}} \sum\limits_{q = p+1}^{{N_a}}  \left| {\left[ {\bf{W}}_{l,i}[n] \right]}_{p,q} \right|  \sin ( \theta _{p,q}^{{\bf{W}}_{l,i}[n]}+2\pi \frac{d}{\lambda }(q - p)\frac{H_k}{\sqrt{ {{ || {\bf{q}}_k^{(o)}[n] - {{\bf{u}}_m}|| }^2 + {H_k^2}} }}) \nonumber \\
&\cdot \frac{4\pi d H_k(q - p) }{\lambda \sqrt{ {{ || {{\bf{q}}_k^{(o)}}[n] - {{\bf{u}}_m} || }^2 + {H_k^2}} }^{3}}\cdot ({{\bf{q}}_k^{(o)}}[n] - {{\bf{u}}_m}),
\end{align}
\end{figure*}
\begin{figure*}
\begin{align}\label{SCA:upsilon:z}
\upsilon^{(z)} ({\bf{R}}_l[n],{\bf{q}}_k^{(o)}[n])&= \sum\limits_{p = 1}^{{N_a}} \sum\limits_{q = p+1}^{{N_a}}  \left| {{{\left[ {{\bf{R}}_l}[n] \right]}_{p,q}}} \right|  \sin ( \theta _{p,q}^{{{\bf{R}}_l}[n]}+ 2\pi \frac{d}{\lambda }(q - p)\frac{H_k}{\sqrt{ {{ || {{{\bf{q}}_k^{(o)}}[n] - {{\bf{u}}_m}}|| }^2 + {H_k^2}} }})\nonumber \\
& \cdot \frac{4\pi d H_k(q - p)}{\lambda \sqrt{ {{ || {{{\bf{q}}_k^{(o)}}[n] - {{\bf{u}}_m}}|| }^2 + {H_k^2}} }^{3}} \cdot ({{\bf{q}}_k^{(o)}}[n] - {{\bf{u}}_m}),
\end{align}
\end{figure*}

Therefore, the UAV trajectory optimization and correspondingly problems (P3) and (P4) can be efficiently solved.

\section{Numerical Results}\label{sec:numerical}

In this section, we provide numerical results to validate the performance of our proposed designs in networked ISAC system for supporting LAE in 3D airspace.

\subsection{Benchmark Schemes}
For comparison, we consider the following benchmark schemes.

\begin{itemize}
\item {\bf{Transmit beamforming with straight flight}}: This scheme assumes that each authorized UAV $k$ takes off from the initial location ${\bf{q}}_k^{\rm{I}}$ and lands at the final location ${\bf{q}}_k^{\rm{F}}$ with straight flight trajectory. The flying speed is assumed to be constant as
\begin{align}
{V_k} = \frac{1}{N}\left\| {{\bf{q}}_k^{\rm{I}} - {\bf{q}}_k^{\rm{F}}} \right\|.
\end{align}
Hence, the whole UAV trajectories ${\bf{q}}_k[n]$ are accordingly determined by
\begin{align}
{{\bf{q}}_k}[n + 1] &= {{\bf{q}}_k}[n] + \frac{{{\bf{q}}_k^{\rm{F}} - {\bf{q}}_k^{\rm{I}}}}{{\left\| {{\bf{q}}_k^{\rm{F}} - {\bf{q}}_k^{\rm{I}}} \right\|}}({V_k}\Delta t), \nonumber \\
&~~~~n \in \{1, \ldots,  N-1\}.
\end{align}
With given $\{{{\bf{q}}_k}[n]\}$, the transmit beamforming $\{{{\bf{w}}_{l,i}}[n]\}$, $\{{{\bf{R}}_{l}}[n]\}$, and UAV-GBS association $\{\alpha_{l,i}[n]\}$ are jointly optimized by solving (P1)-(P4) under different cases.

\item {\bf{Joint power allocation and trajectory design with isotropic transmission}}: In this scheme, the GBSs employ the isotropic transmission with ${\tilde{\bf{W}}_{l,i}}[n] = {\textstyle{{p_{l,i}^c[n]} \over {{N_a}}}}{\bf{I}}$ and ${\tilde{\bf{R}}_l}[n] = {\textstyle{{p_l^s[n]} \over {{N_a}}}}{\bf{I}}$, where ${p_{l,i}^c}[n]$ and ${p_l^s}[n]$ denote the transmit power of information signals and dedicated sensing signals, respectively. Accordingly, the power constraint at each GBS $l$ becomes $\sum\nolimits_{i \in {\cal K}} {p_{l,i}^c[n]}  + p_l^s[n] \le {P_{\max }}$. By substituting $\{{\tilde{\bf{W}}_{l,i}}[n]\}$ and $\{{\tilde{\bf{R}}_l}[n]\}$ into (P1)-(P4) for the four cases, we obtain the corresponding power allocation problems, which can be solved by iteratively optimizing the UAV-GBS association $\{\alpha_{l,i}[n]\}$, power allocation ${p_l^c[n]}$, ${p_l^s[n]}$, and trajectory design $\{{{\bf{q}}_k}[n]\}$.

\end{itemize}

\subsection{Simulation Results}

In the simulation, we consider a networked ISAC system with $M=3$ GBSs located in an area with an acreage of $400~{\text{m}} \times 400~ {\text{m}}$. The antenna spacing of each GBS is set as $d = \frac{\lambda }{2}$ and the number of antennas at each GBS is $N_a=4$. The maximum transmit power budget of each GBSs is $P_{\max}=3 ~{\text{W}}$. In the airspace, we consider $K=2$ authorized UAVs flying at an altitude of ${H_1}={H_2}= 80~{\text{m}}$. The target sensing area is shown in Figs. \ref{fig:bpg} - \ref{fig:bpg:bm} with $Q=20$ sample locations. We set the initial locations of UAVs as ${\bf{q}}_1^{\rm{I}} = [50{\text{m}},250{\text{m}}]$ and  ${\bf{q}}_2^{\rm{I}} = [50 {\text{m}},150{\text{m}}]$ and the final locations as ${\bf{q}}_1^{\rm{F}} = [350{\text{m}},250{\text{m}}]$ and ${\bf{q}}_2^{\rm{F}} = [350{\text{m}},150{\text{m}}]$. The maximum flying speed of authorized UAVs is $V_{\max}=10~ {\text{m/s}}$ and the total number of time slots is $N=40$. Furthermore, the channel power gain at reference distance of one meter is $\kappa=-45~ {\text{dB}}$, and the noise power at authorized UAVs is $\sigma^2=-100~{\text{dBW}}$.

\begin{figure*}[htbp]
\centering
\subfigure[$\Gamma$=-20 dBW, {\textit{Case 1}}.] {\includegraphics[width=4.4cm]{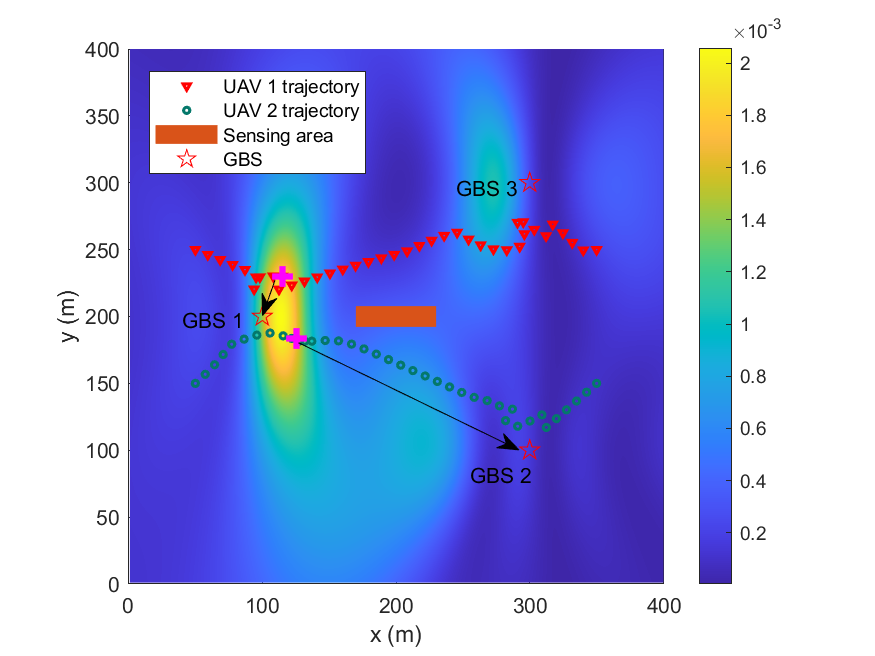}}
\subfigure[$\Gamma$=-7 dBW, {\textit{Case 1}}.]
{\includegraphics[width=4.4cm]{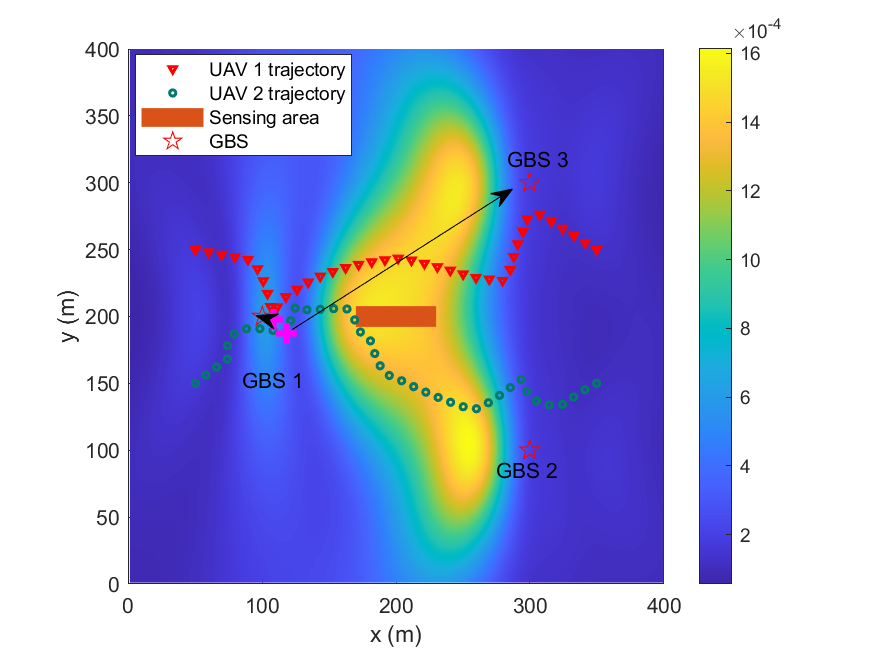}}
\subfigure[$\Gamma$=-20 dBW, {\textit{Case 3}}.]
{\includegraphics[width=4.4cm]{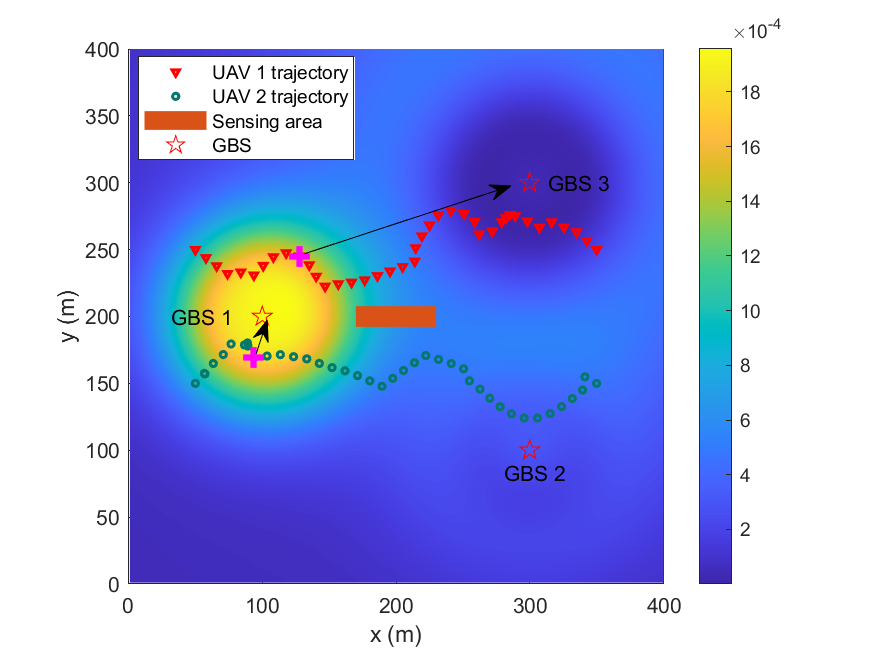}}
\subfigure[$\Gamma$=-7 dBW, {\textit{Case 3}}.]
{\includegraphics[width=4.4cm]{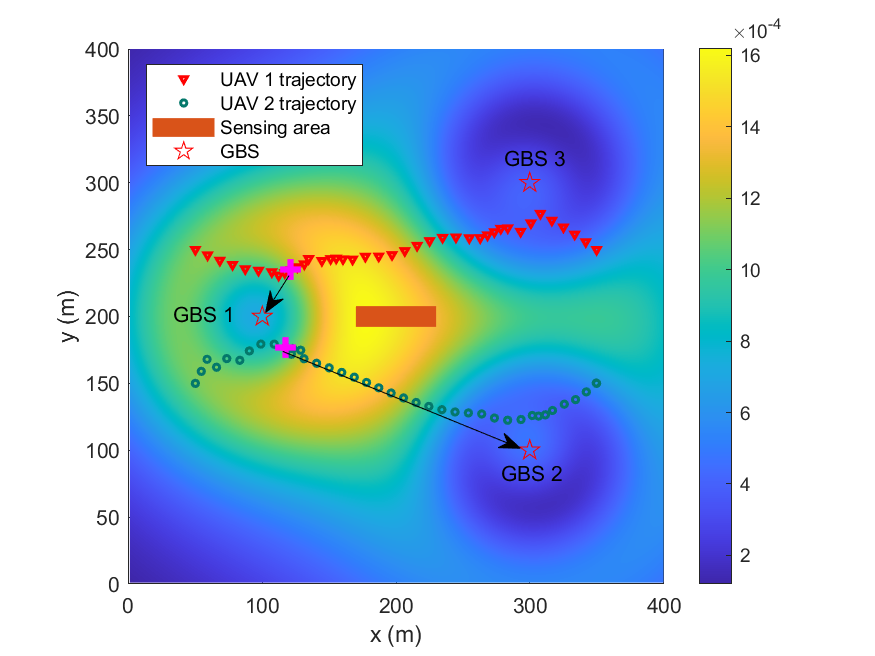}}
\caption{The achieved illumination power gains and authorized UAV trajectory via proposed design with Type-I UAV receivers. The chosen time slot is $N=10$, the carmine '+' denote the authorized UAVs location of this time slot, and the arrow associated with each authorized UAV and GBS correspond to their association relationship.}
\label{fig:bpg}
\end{figure*}

\begin{figure*}[htbp]
\centering
\subfigure[$\Gamma$=-20 dBW, {\textit{Case 2}}.] {\includegraphics[width=4.4cm]{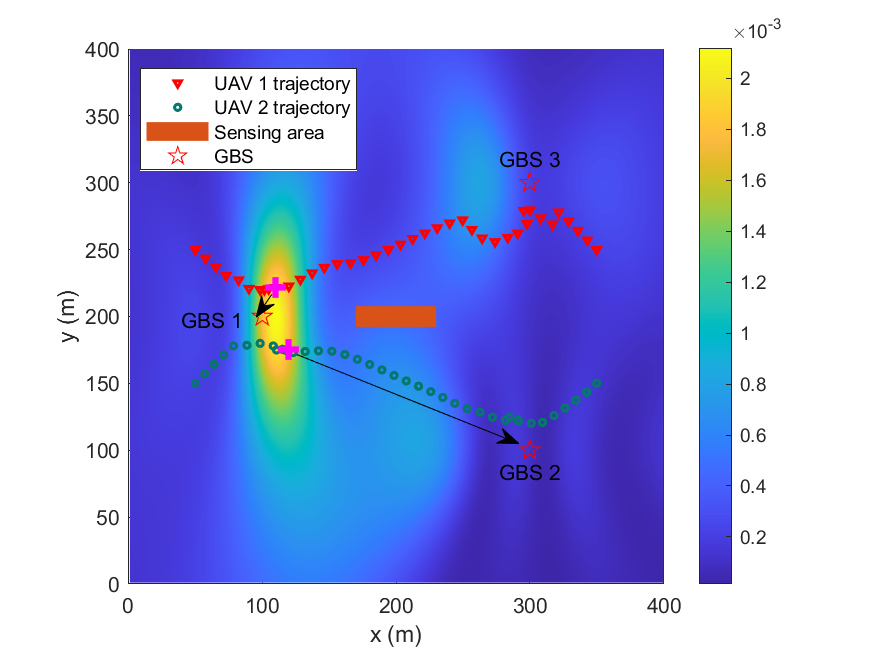}}
\subfigure[$\Gamma$=-7 dBW, {\textit{Case 2}}.]{\includegraphics[width=4.4cm]{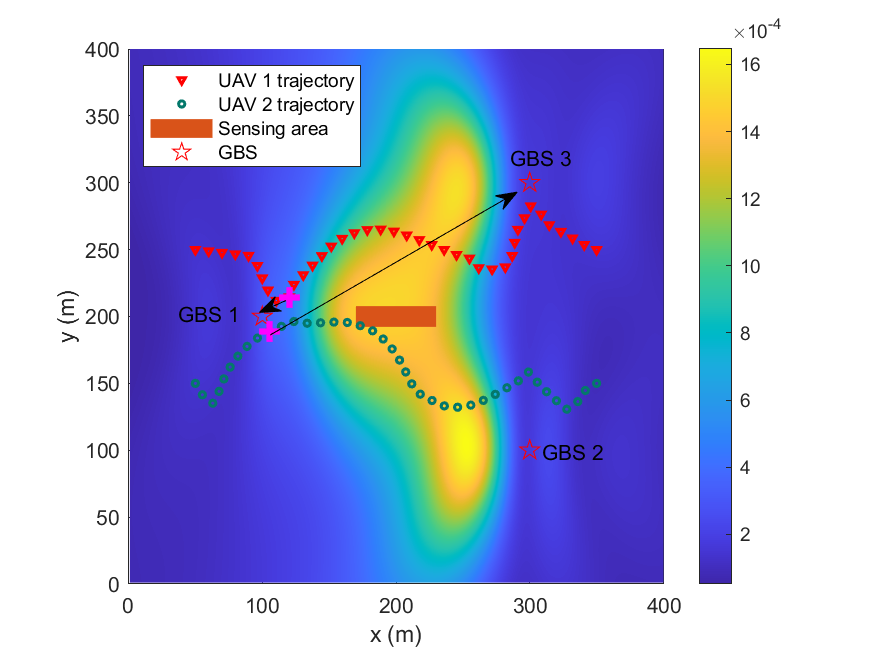}}
\subfigure[$\Gamma$=-20 dBW, {\textit{Case 4}}.]{\includegraphics[width=4.4cm]{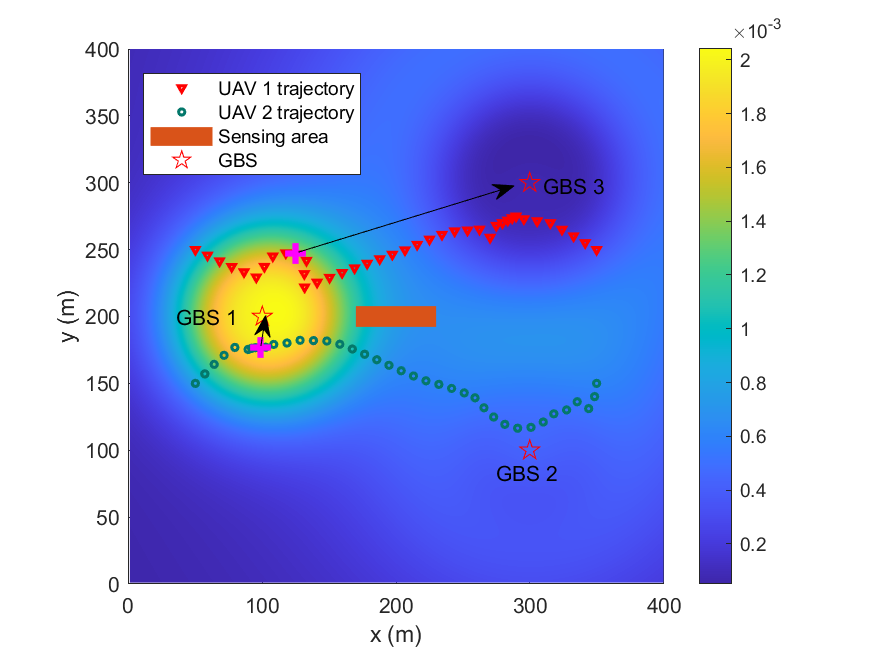}}
\subfigure[$\Gamma$=-7 dBW, {\textit{Case 4}}.]{\includegraphics[width=4.4cm]{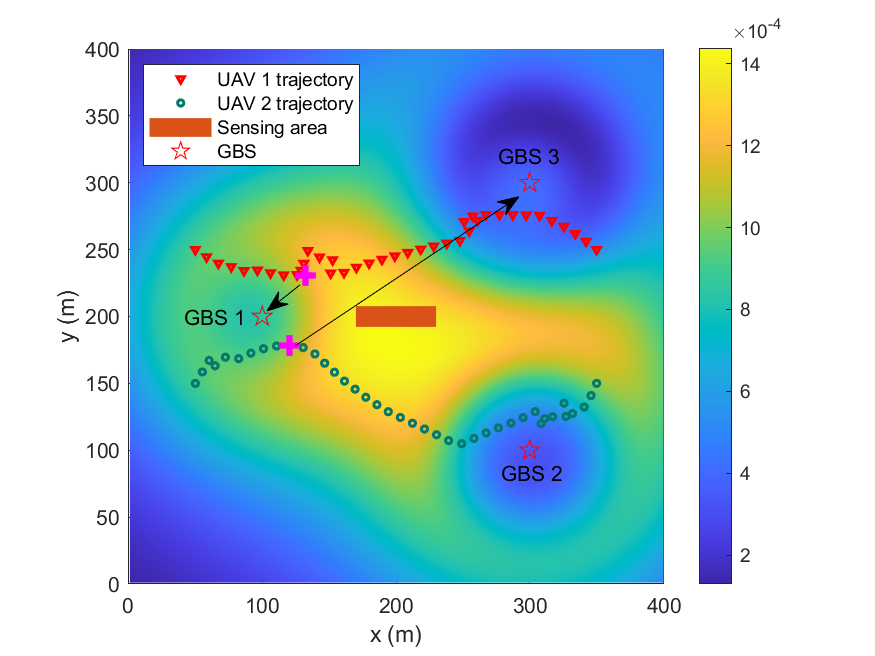}}
\caption{The achieved illumination power gains and authorized UAV trajectory via proposed design with Type-II UAV receivers. The chosen time slot is $N=10$, the carmine '+' denote the authorized UAVs location of this time slot, and the arrow associated with each authorized UAV and GBS correspond to their association relationship.}
\label{fig:bpg:II}
\end{figure*}

\begin{figure*}[htbp]
\centering
\subfigure[Transmit beamforming with straight flight, $\Gamma$=-20 dBW, {\textit{Case 1}}.] {\includegraphics[width=4.4cm]{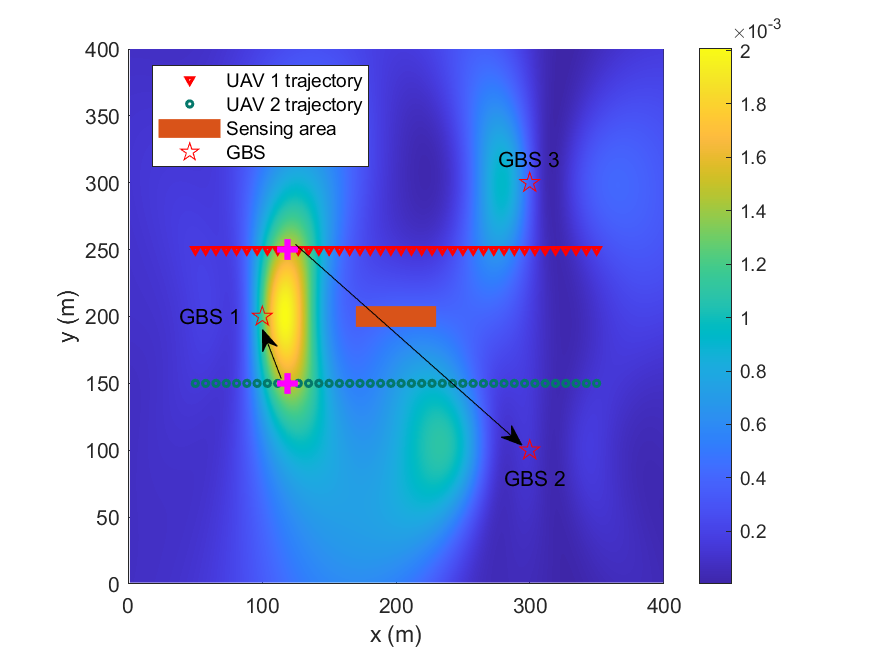}}
\subfigure[Transmit beamforming with straight flight, $\Gamma$=-7 dBW, {\textit{Case 1}}.] {\includegraphics[width=4.4cm]{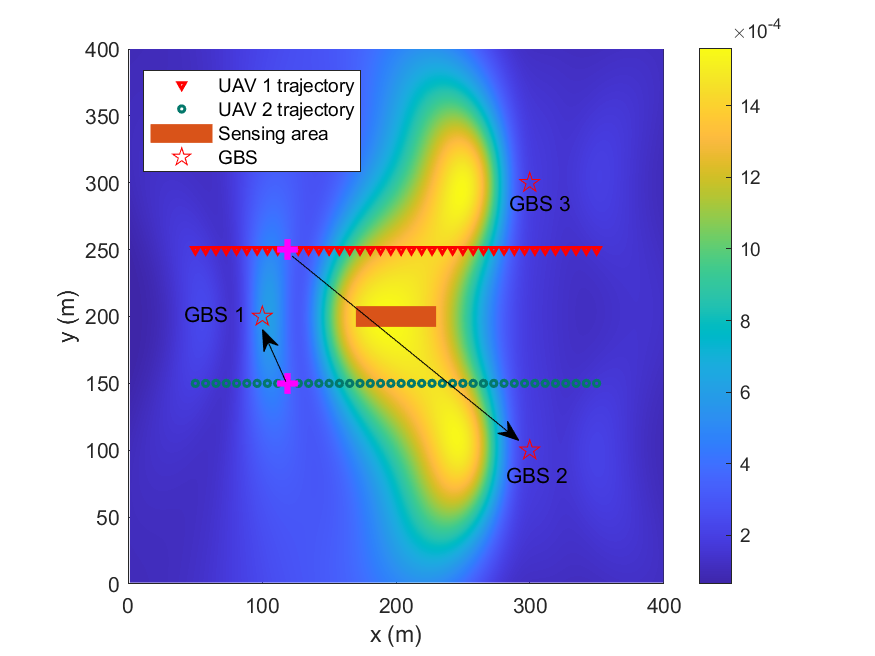}}
\subfigure[Joint power allocation and trajectory design, $\Gamma$=-20 dBW, {\textit{Case 1}}.] {\includegraphics[width=4.4cm]{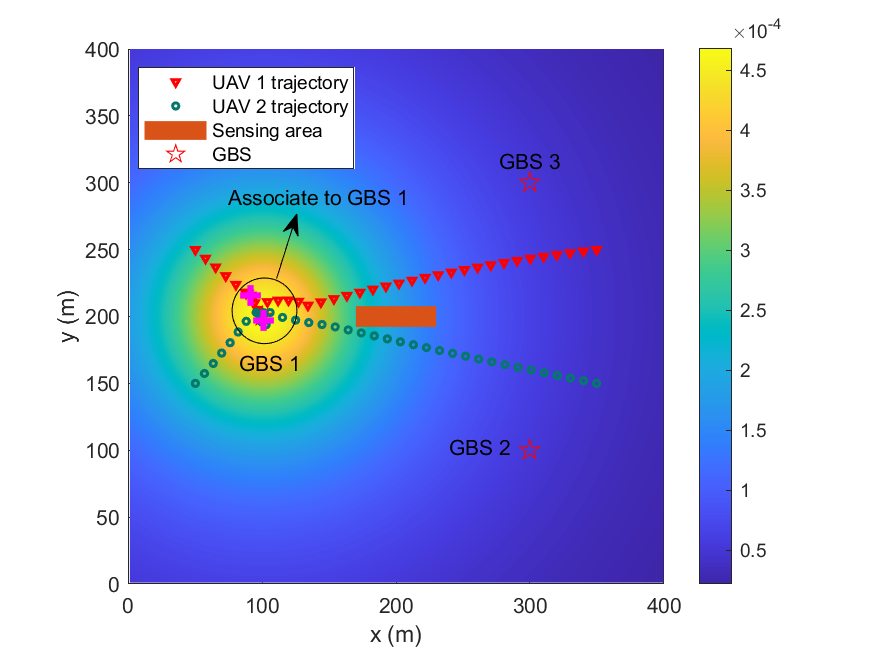}}
\subfigure[Joint power allocation and trajectory design, $\Gamma$=-13 dBW, {\textit{Case 1}}.] {\includegraphics[width=4.4cm]{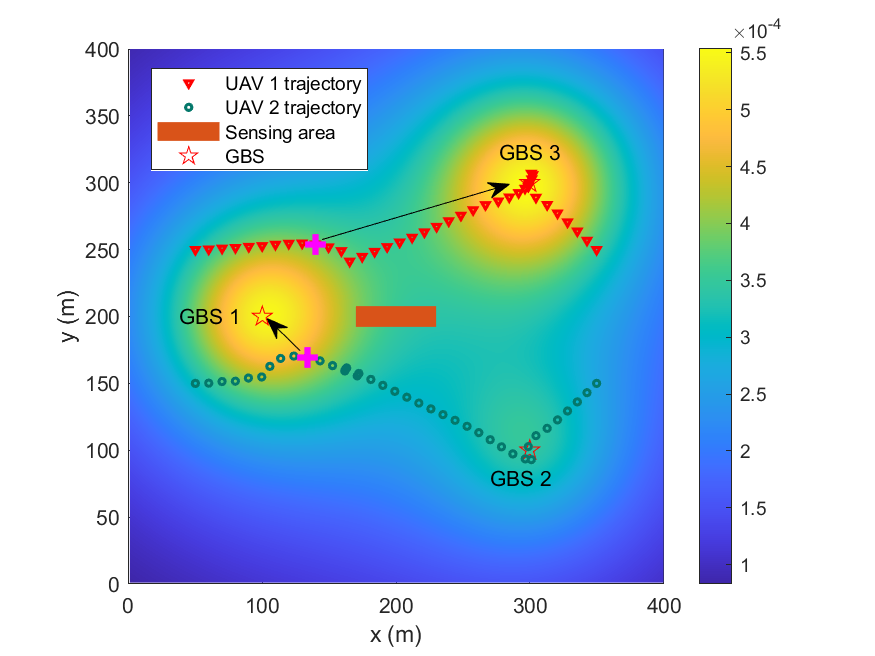}}
\caption{The achieved illumination power gains and UAV trajectory via benchmark designs of Transmit beamforming with straight flight and Joint power allocation and trajectory design under {\textit{Case 1}}. The chosen time slot is $N=10$, the carmine '+' denote the UAVs location of this time slot, and the arrow associated with each UAV and GBS correspond to their association relationship.}
\label{fig:bpg:bm}
\end{figure*}

Fig. \ref{fig:bpg} shows the optimized illumination power gains and authorized UAV trajectories with Type-I UAV receivers by our proposed design at time slot $n=10$. It is observed that under both horizontal and vertical antenna configurations, as $\Gamma$ increases from $-20$ dBm to $-7$ dBm, the authorized UAVs fly more closely to the GBSs to obtain enhanced communication rate, and the illumination power is more concentrated towards the sensing area. This is because when the sensing constraints become stringent, the GBSs need to steer the transmit beams toward the sensing area at each time slot, and thus the authorized UAVs prefer to fly closer to the GBSs to exploit the strong signal power. Comparing Figs. \ref{fig:bpg}(a) and \ref{fig:bpg}(b) versus Figs. \ref{fig:bpg}(c) and \ref{fig:bpg}(d), it is observed that the schemes with horizontal antennas achieve higher illumination power gains at the authorized UAVs than their counterparts with vertical antennas. This is because the vertically placed antennas can only change its beamformers vertically to serve the authorized UAVs in 3D airspace, which severely restricts the ISAC coverage performance. By contrast,  the horizontal antennas can properly adjust the fan-shaped mainlobe to achieve enhanced 3D coverage. It is also observed that in Fig. \ref{fig:bpg}(a) and \ref{fig:bpg}(b), UAV $1$ is associated to GBS $1$, thus hovering near GBS $1$ for higher communication performance, while UAV $2$ is associated to GBS $2$ when $\Gamma=-20$ dBW and associated to GBS $3$ when $\Gamma=-7$ dBW. The corresponding trajectories of UAV $1$ and UAV $2$ show that the directions of flight are also towards the associated GBS $2$ and GBS $3$. Similar observations can be found in Fig. \ref{fig:bpg}(c) and \ref{fig:bpg}(d).

Fig. \ref{fig:bpg:II} depicts the optimized illumination power gains and authorized UAV trajectories with Type-II receivers by our proposed schemes at time slot $n=10$. By comparing Fig. \ref{fig:bpg:II} versus Fig. \ref{fig:bpg}, it is observed that the illumination power gains with Type-II UAV receivers are higher than that with Type-I UAV receivers. This validates the importance of canceling dedicated sensing signal interference in enhancing ISAC performance. Furthermore, the interference-canceling ability of GBSs also influences the optimized UAV trajectories. It is observed that the two UAVs in Fig. \ref{fig:bpg:II} keep a greater distance especially from $n=10$ to $n=30$ than that in Fig. \ref{fig:bpg} when flying through the airspace that is near the sensing area. This is due to the fact that the authorized UAVs with Type-II receivers have the ability to cancel the sensing interference, and thus GBSs can steer the information and sensing beams more concentratedly to the sensing area even when UAVs are flying near it. By contrast, when GBSs serve Type-I UAV receivers, they tend to reduce the sensing beam strength to avoid the sensing interference that cannot be cancelled. This thus strengthens the communication beams to illuminate sensing area, which allows Type-I UAV receivers to fly closer to the sensing area.


Fig. \ref{fig:bpg:bm} shows the optimized UAV trajectory and illumination power gain by the two benchmarks, i.e., transmit beamforming with straight flight and joint power allocation and trajectory design with isotropic transmission under {\textit{Case 1}}. In Figs. \ref{fig:bpg:bm}(a) and \ref{fig:bpg:bm}(b), the UAV trajectories are fixed as straight lines, which result in limited communication performance, as the UAVs are not allowed to fly near to GBSs for communication. The illumination power gain at the sensing area is also degraded because the GBSs need to steer the beams to illuminate the pre-designed routine, while the UAV trajectories can be designed to be close to the sensing area to facilitate both sensing and communication in the proposed designs. Thanks to the coordinated transmit beamforming design at GBSs, the degradation of sensing performance is less significant. In Figs. \ref{fig:bpg:bm}(c) and \ref{fig:bpg:bm}(d), the illumination power gain by isotropic transmission is significantly worse than the schemes with beamforming optimization. This is due to the fact that the GBSs only perform the power allocation design, which lacks the DoFs for communicating with UAVs and sensing the interested area. Thus leading to worse interference suppression. The above results validate again the importance of our joint coordinated transmit beamforming and UAV trajectory design.

Fig. \ref{fig:tradeoff} shows the achieved average sum rate of authorized UAVs versus the illumination power for sensing. It is observed that the UAVs' average sum rate declines for all schemes as the sensing requirement $\Gamma$ increases. This is because GBSs need to spend more transmit power to satisfy the sensing requirements by covering the interested sensing area. It is also observed that the proposed scheme under {\textit{Case 3}} achieves the highest average sum rate among all the considered schemes, and the performance gap over other schemes becomes larger when $\Gamma$ increases. This validates again the benefit of the proposed joint coordinated transmit beamforming and UAV trajectory design, and shows that the horizontally placed antennas are beneficial for LAE. Furthermore, the benchmark of joint power allocation and trajectory design is observed to perform significantly worse than other schemes, and even becomes infeasible when $\Gamma > 13$ dBW. This is due to its limited ability to reshape the transmit beamformers to cater to increased sensing requirements.


\begin{figure}[ht]
\centering
    \includegraphics[width=8cm]{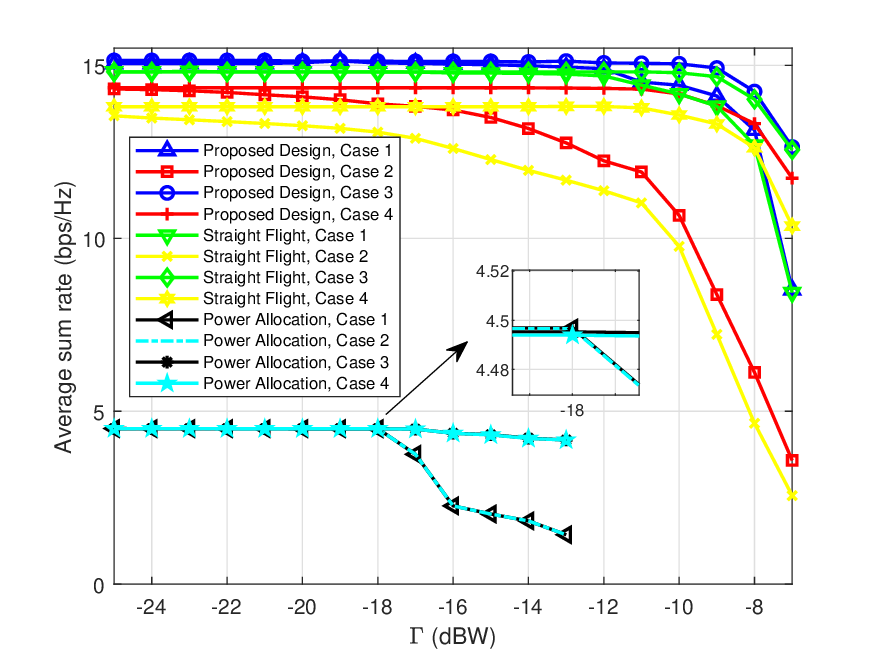}
\caption{The average sum rate of authorized UAVs versus the illumination power for sensing.}
\label{fig:tradeoff}
\end{figure}


\section{Conclusion} \label{sec:conclusion}
This paper considered a networked ISAC system, in which multiple GBSs employed joint information and sensing signals to perform cooperative sensing toward an targeted 3D area and communicate with multiple UAVs to support LAE. We considered two configurations of horizontal and vertical antennas at GBSs. We also considered two types of UAV receivers that without and with the ability to cancel the dedicated sensing interference. Under each setup, we proposed the joint coordinated transmit beamforming and trajectory design to maximize the average sum rate of authorized UAVs, subject to the illumination power gain constraints toward the interested sensing area, the transmit power constraints at GBS, and the practical UAV flight constraints. We proposed efficient algorithms to solve the resultant mixed-integer non-convex problems by using the AO, SDR, and SCA techniques. Finally, numerical results showed that the proposed joint transmit beamforming and trajectory design with Type-II UAV receivers and horizontal antennas at GBSs significantly outperforms other designs in supporting LAE. The UAV trajectory design significantly enhances the communication performance with properly controlled interference and the coordinated transmit beamforming together with  horizontal antennas at GBSs consistently lead to better ISAC performance.


\appendix
\subsection{Proof of Proposition 1} \label{appendixA}
In the proof, we omit the superscript ($x$) and ($o$) in the following formulas for brevity. It can be verified that based on \eqref{app:1:1} and \eqref{app:1:3}, the constructed solutions of $ {{\bf{\bar W}}}_{l,i}[n]$ and $ {{\bf{\bar R}}}_{l}[n]$ satisfy the illumination power gain constraint \eqref{P2:con:sen} and the power constraint \eqref{P2:con:pow} in problem (P1.3.$o$).

Next, we verify that $ {{\bf{\bar W}}}_{l,i}[n]$ and $ {{\bf{\bar R}}}_{l}[n]$ also achieve the same objective for problem (P1.3.$o$). From \eqref{app:1:1} and \eqref{app:1:2}, we have
\begin{align}
&{\rm{tr}}({\bf{H}}_{m}({{\bf{q}}_k}[n],H_k){{{\bf{\bar W}}}_{l,i}}[n]) \nonumber \\
&= {\rm{tr}}({\bf{h}}_{m}^H({{\bf{q}}_k}[n],H_k){{\bf{\bar w}}_{l,i}} {\bf{\bar w}}_{l,i}^H[n]{\bf{h}}_{m} ({{\bf{q}}_k}[n],H_k)) \nonumber \\ &= {\rm{tr}}({{\bf{H}}_{m}}({{\bf{q}}_k}[n],H_k){\bf{W}}_{l,i}^ * [n]). \label{app:1:4}
\end{align}
For the first term in \eqref{rate:SCA}, it follows that
\begin{align}
&{\log _2}(\sum\limits_{l \in {\cal M}} {{\rm{tr}}({{\bf{H}}_{m}}({{\bf{q}}_k}[n],H_k)(\sum\limits_{i \in {\cal K}} {{{{\bf{\bar W}}}_{l,i}}[n]}  + {{{\bf{\bar R}}}_l}[n]}) ) + {\sigma ^2}) \nonumber \\
&={\log _2}(\sum\limits_{l \in {\cal M}} {{\rm{tr(}}{{\bf{H}}_{m}}({{\bf{q}}_k}[n],H_k)(\sum\limits_{i \in {\cal K}} {{\bf{ W}}_{l,i}^ * [n]}  + {\bf{ R}}_l^ * [n])} ) + {\sigma ^2}). \label{app:1:5}
\end{align}
For the remaining terms of \eqref{rate:SCA}, we have
\begin{subequations}\label{app:1:6}
\begin{align}
&{a_{m,k}^{(\text{I},o)}}[n] \!\! +\!\! \sum\limits_{(l,i) \ne (m,k)} \!\!{\rm{tr}} ({\bf{B}}_{m}^{(\text{I},o)} ({{\bf{q}}_k}[n],H_k) \! \cdot \! ({{{\bf{\bar W}}}_{l,i}}[n] \!-\! {\bf{W}}_{l,i}[n])) \nonumber \\
&+\sum\limits_{l \in {\cal M}} {{\rm{tr}}} ({\bf{B}}_{m}^{(\text{I},o)}({{\bf{q}}_k}[n],H_k) \cdot ({{{\bf{\bar R}}}_l}[n] - {\bf{R}}_l[n])) \nonumber \\
&= {a_{m,k}^{(\text{I},o)}}[n]  +  \frac{{{\log }_2}(e)}{2^{a_{m,k}^{(\text{I},o)}[n]}}  \sum\limits_{l \in {\cal M}} {\rm{tr}}({\bf{H}}_{m}({{\bf{q}}_k}[n],H_k)\nonumber \\
& \cdot (\sum\limits_{i \in {\cal K}} {\bf{\bar W}_{l,i}}[n]  + {{\bf{\bar R}}_l}[n] )) \nonumber \\
& - \frac{{{\log }_2}(e)}{2^{{a_{m,k}^{(\text{I},o)}}[n] }} \sum\limits_{l \in {\cal M}} {\rm{tr}}({{\bf{H}}_{m}}({{\bf{q}}_k}[n],H_k) \cdot (\sum\limits_{i \in {\cal K}} {{\bf{W}}_{l,i}[n]}  + {{\bf{R}}_l}[n] )) \nonumber \\
&- \frac{{{\log }_2}(e)}{2^{{a_{m,k}^{(\text{I},o)}}[n]}}{\rm{tr}} ({\bf{H}}_{m}({{\bf{q}}_k}[n],H_k) \cdot ({{\bf{\bar W}}_{m,k}}[n] - {\bf{W}}_{m,k}[n])) \nonumber \\
&= {a_{m,k}^{(\text{I},o)}}[n]  +  \frac{{{\log }_2}(e)}{2^{a_{m,k}^{\text{I}}[n]}}  \sum\limits_{l \in {\cal M}} {\rm{tr}}({\bf{H}}_{m}({{\bf{q}}_k}[n],H_k)\nonumber \\
& \cdot (\sum\limits_{i \in {\cal K}} {\bf{ W}_{l,i}^{*}}[n]  + {{\bf{ R}}_l^{*}}[n] )) \nonumber \\
& - \frac{{{\log }_2}(e)} {2^{{a_{m,k}^{(\text{I},o)}}[n]}} \sum\limits_{l \in {\cal M}} {\rm{tr}}({{\bf{H}}_{m}}({{\bf{q}}_k}[n],H_k) \cdot (\sum\limits_{i \in {\cal K}} {{\bf{W}}_{l,i}[n]}  + {{\bf{R}}_l}[n] ))\nonumber \\
&- \frac{{{\log }_2}(e)}{2^{{a_{m,k}^{(\text{I},o)}}[n]}}{\rm{tr}} ({\bf{H}}_{m}({{\bf{q}}_k}[n],H_k) \cdot ({{\bf{ W}}_{m,k}^{*}}[n] - {\bf{W}}_{m,k}[n])) \label{app:1:8}\\
&={a_{m,k}^{(\text{I},o)}}[n] + \sum\limits_{(l,i) \ne (m,k)} \frac{{{{\log }_2}(e)}}{2^{{a_{m,k}^{(\text{I},o)}}[n]}}{\rm{tr}} ({{\bf{H}}_{m}}({{\bf{q}}_k}[n],H_k) \nonumber\\
& \cdot ({\bf{W}}_{l.i}^ * [n] - {{\bf{W}}_{l,i}}[n])) \nonumber\\
&+ \sum\limits_{l \in {\cal M}} {\frac{{{{\log }_2}(e)}}{{{2^{{a_{m,k}^{(\text{I},o)}}[n]}}}}{\rm{tr}}} ({{\bf{H}}_{m}}({{\bf{q}}_k}[n],H_k) \cdot({\bf{R}}_l^ * [n] - {{\bf{R}}_l}[n])) \nonumber \\
&={a_{m,k}^{(\text{I},o)}}[n] \!\! + \!\! \! \sum\limits_{(l,i) \ne (m,k)} \!\! \!\! {\rm{tr}} ({\bf{B}}_{m}^{(\text{I},o)} ({{\bf{q}}_k}[n],H_k)\! \cdot \! ({\bf{ W}_{l,i}^{*}}[n] \!\!-\!\! {\bf{W}}_{l,i}[n]))\nonumber \\
&+\sum\limits_{l \in {\cal M}} {{\rm{tr}}} ({\bf{B}}_{m}^{(\text{I},o)}({{\bf{q}}_k}[n],H_k) \cdot ({\bf{R}_l^{*}}[n] - {\bf{R}}_l[n])). \nonumber \label{app:1:10}
\end{align}
\end{subequations}
In \eqref{app:1:6}, the equality \eqref{app:1:8} follows from \eqref{app:1:3} and \eqref{app:1:4}. It is clear that, by combining \eqref{app:1:5} and \eqref{app:1:6}, it holds that the objective values of (P1.3.$o$) and (SDR1.3.$o$) remain the same. According to the above derivation, the constructed solution of $ {{\bf{\bar W}}}_{l,i}[n]$ and $ {{\bf{\bar R}}}_{l}[n]$ is also the optimal solution to (P1.3.$o$). Thus, this completes the proof of  {\textit{Proposition 1}}.

\bibliographystyle{IEEEtran}
\bibliography{IEEEabrv,myref}

\end{document}